\documentclass[runningheads]{llncs}
\usepackage[T1]{fontenc}

\usepackage{graphicx}
\usepackage{amsfonts}
\usepackage{amsmath}
\usepackage{amssymb}
\usepackage{enumitem}
\usepackage{tabularx}
\usepackage{array}
\usepackage{wrapfig}
\usepackage{fancyvrb}
\usepackage{caption}
\usepackage{xspace}
\usepackage[title]{appendix}
\graphicspath{{figures/}}
\usepackage{pifont}

\usepackage{xcolor}
\usepackage[export]{adjustbox}
\usepackage{placeins}
\usepackage{a4wide} %if we use style llncs, to make the text box wider
\usepackage{hyperref}
\usepackage{float}
\usepackage{booktabs}
\newcommand{\centercell}[1]{\multicolumn{1}{c}{#1}}

\newcommand{\blue}[1]{{\textcolor{black}{#1}}}

\newcommand{\single}[1]{{\textcolor{blue!80!green!80!black}{#1}}}
\newcommand{\multi}[1]{{\textcolor{blue!40!green!80!black}{#1}}}

\newcolumntype{Y}{>{\centering\arraybackslash}X} %%%%% Removed
\newcommand{\etal}{et al.\xspace}
\newcommand{\eps}{\varepsilon}

\newcommand{\new}[1]{{#1}}

\newtheorem {observation}[theorem] {Observation}

\begin{document}
\title{Terrain prickliness: theoretical grounds\\ for high complexity viewsheds\thanks{A preliminary version of this paper appeared in Proc. 11th International Conference on Geographic Information Science (GIScience'21), Part II, 10:1-10:16, 2021.}}
\author{Ankush Acharyya\inst{1} \and
Maarten L\"{o}ffler\inst{2,3} \and 
Gert G.T. Meijer\inst{4} \and 
Maria Saumell\inst{5} \and 
Rodrigo I. Silveira\inst{6} \and 
Frank Staals\inst{2} }
\authorrunning{Acharyya et al.}
\institute{Department of Computer Science and Engineering, National Institute of Technology, Durgapur %, Czech Republic 
%\email{\{acharyya,jallu,saumell\}@cs.cas.cz} 
\and
Dept. of Information and Computing Sciences, Utrecht University%, Netherlands 
%\email{m.loffler@uu.nl, g.t.meijer@students.uu.nl} 
\and
Department of Computer Science, Tulane University
\and
Academy of ICT and Creative Technologies, NHL Stenden University of Applied Sciences
\and 
Dept. of Theoretical Computer Science, Faculty of Information Technology, Czech Technical University in Prague %, Czech Republic 
\and 
Dept. de Matem\`{a}tiques, Universitat Polit\`{e}cnica de Catalunya%, Spain %\email{rodrigo.silveira@upc.edu} 
}

\maketitle   % typeset the header of the contribution
%
%\linenumbers

\begin{abstract}
An important task in terrain analysis is computing
  \emph{viewsheds}. A viewshed is the union of all the parts of the
  terrain that are visible from a given viewpoint or set of
  viewpoints. The complexity of a viewshed can vary significantly
  depending on the terrain topography and the viewpoint position. In
  this work we study a new topographic attribute, the
  \emph{prickliness}, that measures the number of local maxima in a
  terrain from all possible angles of view. We show that the
  prickliness effectively captures the potential of 2.5D TIN terrains to have high complexity viewsheds.
  \new{We present
  optimal and (under standard assumptions) near-optimal 
  algorithms to compute it for 1.5D and 2.5D TIN terrains, respectively, and efficient approximate algorithms for raster DEMs. }
 % We present optimal (for 1.5D terrains) and near-optimal (for 2.5D terrains) algorithms to compute it for TIN terrains, and efficient approximate algorithms for raster DEMs. 
  We validate the usefulness of the prickliness attribute with experiments in a large set of real terrains.

\end{abstract}

\section{Introduction}

%\maria{This is a new version of Sections 1 and 2}

Digital terrain models represent part of the earth's surface, and are used to solve a variety of %analysis
problems in Geographic Information Science (GIS).
An important task %for which terrain models are used 
is viewshed analysis: determining which parts of a terrain are visible from certain terrain locations.
%Here two points on the terrain are considered visible if the straight line segment connecting them does not intersect the interior of the terrain.
Two points $p$ and $q$ on or above a terrain are mutually \emph{visible} if the line of sight defined by line segment $\overline{pq}$ does not intersect the interior of the terrain.
%The most fundamental structure in this context is the set of all points that are visible from a given \emph{viewpoint} $p$.
%This is known as the \emph{viewshed} of $p$.
Given a \emph{viewpoint} $p$, the \emph{viewshed} of $p$ is the set of all terrain points that are visible from $p$. Similarly, the viewshed of a set of viewpoints $P$ is defined as the set of all terrain points that are visible from \emph{at least} one viewpoint in $P$.
Viewsheds %are fundamental structures  that 
are useful, for example, in evaluating the visual impact of potential constructions~\cite{Danese2011short}, analyzing the coverage of an area by fire watchtowers~\cite{klms-papgpt-14}, or measuring the scenic beauty of a landscape~\cite{CHAMBERLAIN201313,SCHIRPKE20131}. %R: removed citation to Kucuk17 to shorten

\subsection{Discrete and continuous terrain representations}
Two major terrain representations are prevalent in GIS.
The simplest and most widespread is the raster, or \emph{digital elevation model} (DEM), consisting of a rectangular grid where each cell stores an elevation.\footnote{For the sake of simplicity, in this paper we use DEM to denote the raster version of a DEM.}
The main alternative is a vector representation, or \emph{triangulated irregular network} (TIN), where a set of irregularly spaced elevation points are connected into a triangulation. %, which provides a piece-wise linear interpolating function.
A TIN can be viewed as a continuous $xy$-monotone polyhedral surface in $\mathbb{R}^3$.
\new{
Following standard terminology, in this paper we will refer to such terrain representations as \emph{2.5D terrains}.
The \emph{2.5D} part (as opposed to \emph{3D}) arises from the fact that what is represented is an $xy$-monotone surface in 3D, but it is not fully three-dimensional (since, for instances, caves and overhangs cannot be represented in these models).}

A viewshed in a DEM is the set of all raster cells that are visible from at least one viewpoint.
In contrast, a viewshed in a TIN is the union of all
\emph{parts of triangles} that are visible from at least one
viewpoint, so it can be seen as a set of polygons.

DEMs are simpler to analyze than TINs and facilitate most analysis tasks.
The main advantage of TINs is that they require less storage space.
Both models have been considered extensively in the literature for viewshed analysis, see Dean~\cite{Dean97} for a complete comparison of both models in the context of forest viewshed.
Some studies suggest that TINs can be superior to DEMs in viewshed computations~\cite{Dean97}, but experimental evidence is inconclusive~\cite{RD-CausesError-07}.
This is in part due to the fact that the viewshed algorithms used in~\cite{RD-CausesError-07} do not compute the visible part of each triangle, but only attempt to  determine whether each triangle is completely visible.
This introduces an additional source of error and does not make use of all the information contained in the TIN.

%\begin{wrapfigure}{r}{0.33\textwidth}
\begin{figure}[h]
\begin{center}
    \includegraphics[width=0.32\textwidth]{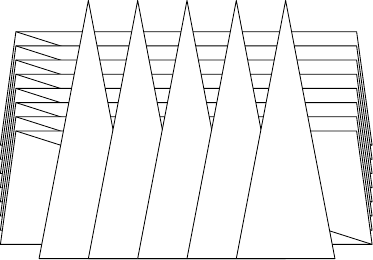}
  \end{center}

  \caption{Part of a TIN with a high-complexity viewshed. The viewpoint (not shown) is placed at the center of projection. The relevant triangles of the TIN are the ones shown, which define $n$ peaks and ridges. The viewshed in this case is formed by $\Theta(n^2)$ visible regions.}
  \label{fig:quadratic}
\end{figure}
 % \end{wrapfigure}
 
\subsection{Viewshed complexity}
%From a theoretical point of view, 
The algorithmic study of viewsheds focuses on two main aspects: the complexity of the viewsheds, and their efficient computation.
%In this work, we are interested in their complexity: intuitively, the complexity is the amount of information required to describe a viewshed.
In this work, we are interested in their complexity. We use the information-theoretic meaning of ``complexity'': the complexity of an object is the number of bits needed to represent it in memory.
%, since this directly affects their computation time.
%In a $1.5$D terrain, the viewshed of one viewpoint can have a complexity that is linear on the  number of vertices, and can be computed in linear time~\cite{JoeS87}.
%The complexity of a viewshed is usually measured in terms of the number of maximal visible regions and the complexity of their boundary (i.e., a high-complexity viewshed can be formed by a large number of regions, or by a few larger regions with a very complex boundary).
%In the case of TINs, viewshed complexity is defined simply as the total number of vertices of the polygons that form the viewshed.
Therefore, in the case of TINs, viewshed complexity is defined as the total number of vertices of the polygons that form the viewshed.
In the case of DEMs, there are several ways to measure viewshed complexity. To facilitate comparison between TIN and DEM viewsheds, we convert the visible areas in the raster viewshed to polygons, and define the viewshed complexity as the total number of vertices in those polygons.
%
%The complexity of a viewshed in a DEM can be proportional to the number of pixels in the raster.
%In contrast, in a TIN terrain the complexity of the viewshed of one viewpoint can be quadratic in the number of vertices~\cite{McKenna87}, which makes its computation too slow for most practical uses when terrains are large.
%
%
%
A typical high-complexity viewshed construction for a TIN is shown schematically in Fig.~\ref{fig:quadratic}, where one viewpoint would be placed at the center of projection, and both the number of vertical and horizontal triangles is $\Theta(n)$, for $n$ terrain vertices.
The vertical peaks form a grid-like pattern with the horizontal triangles, leading to a viewshed with $\Theta(n^2)$ visible triangle pieces.
%There exist several algorithms whose running times are proportional to the number of terrain triangles and also to the viewshed size of the terrain~\cite{KatzOS92,ReifS89short}.
%\rodrigo{Commented out sentence about computation, since it looked a bit unrelated to the rest}

% (changed) There are several output-sensitive algorithms in the literature~\cite{KatzOS92,ReifS89}
% available to compute the viewshed for a $2.5$D terrain.

%Apart from single viewpoints, viewshed computation has also been studied for multiple viewpoints~\cite{fishy2014}.
%In this setting, there are $m$ viewpoints placed on the terrain, and the visible parts of the terrain are defined at those visible from \emph{at least one} viewpoint.\rodrigo{Adjust this definition depending on what we need later}

While a viewshed can have high complexity, this is expected to be uncommon in real terrains~\cite{BergHT10}.
%For this reason, 
There have been attempts to define theoretical conditions for a (TIN) terrain 
%to be ``realistic'' 
that guarantee, among others, that viewsheds cannot be that large.
For instance, Moet \etal~\cite{MOET200848} showed that if terrain triangles satisfy certain ``realistic'' shape conditions, viewsheds have $O(n \sqrt{n})$ complexity.
De Berg \etal~\cite{BergHT10} showed that similar conditions  guarantee worst-case expected complexity of $\Theta(n)$ when the vertex heights are subject to uniform noise.

\subsection{Viewsheds and peaks}

%To produce a viewshed, one needs two ingredients: a terrain and a set of viewpoints.
%While both play an important role in the viewshed complexity, it is often the case that multiple viewsheds need to be computed on the same terrain model.
%Consider, for instance, the study of candidate locations for a facility, where a viewshed is needed for each possible location, but always on the same terrain.
%In this work, our goal is to identify terrain features that can lead to high complexity viewsheds, when combined with the right set of viewpoints.

%In the GIS community, it is well-known that the topography of the terrain has a strong influence on the potential complexity of the viewshed (something well-studied for sitting observers on terrains to maximize coverage, see e.g.,~\cite{KIM20041019}).
The topography of the terrain has a strong influence on the potential complexity of the viewshed. 
% (see e.g.~\cite{Kim2004ViewpointLocations} for placing observers to maximize coverage).
To give an extreme example, in a totally
concave %\footnote{Informally, a point $p$ on (the surface of) a terrain will be considered {\em concave} (resp., {\em convex}) if there exists a non-vertical plane through $p$ that leaves all neighboring points above it (resp., below it).}
%concave\footnote{In this work, as is common in terrain analysis but unlike functional analysis, we use {\em convex} and {\em concave} to refer to the earth volume {\em under} the terrain surface, rather than the air volume above. \maria{TO DO: Explain more. One of the referees asked what we meant by a convex vertex}}
terrain, the viewshed of any viewpoint will be the whole terrain, and has a trivial description. \blue{Intuitively, to obtain a high complexity viewshed as in Fig.~\ref{fig:quadratic}, one needs a large number of obstacles obstructing the visibility from the viewpoint, which %in turn 
requires a somewhat rough topography.}

%(changed) To obtain a high complexity viewshed as in the figure above, the intuition is that one needs a large number of obstacles obstructing the visibility from the viewpoint, which in turn requires a somewhat rough topography.

In fact, it is well-established that viewsheds tend to be more complex in terrains that are more ``rugged''~\cite {Kim2004ViewpointLocations}.
%Similarly, viewsheds tend to be
%This has been studied, for instance, in the context of watchtower placement for visual surveillance \cite {franklin,klms-papgpt-14}.
%Several measures of the ruggedness of terrains have been introduced: ...
This leads to the natural question of which terrain characteristics correlate with high complexity viewsheds.
% (e.g., quadratic-size TIN viewsheds).
Several topographic attributes have been proposed to capture different aspects of the roughness of a terrain, such as  the \emph{terrain ruggedness index}~\cite{Riley1999TRI}, the \emph{terrain shape index}~\cite{McNab1989TSI}, or the \emph{fractal dimension}~\cite{Mandelbrot1982}.
These attributes focus on aspects like the amount of elevation change between adjacent parts of a terrain, its overall shape, or the terrain complexity. However, none of them is specifically intended to capture the possibility to produce high complexity viewsheds, and there is no theoretical evidence for such a correlation.
Moreover, these attributes are locally defined,
% for one single point, 
and measure only attributes of the local neighborhood of one single point. 
While we can average these measures over the whole terrain,
given the global nature of visibility, it is unclear a priori whether such measures are suitable for predicting viewshed complexity.
%While one can measure it for every single terrain point and compute, for instance, the average value (as we do in our experiments), these measures are still local rather than global.
We refer to Dong \etal~\cite{Dong2008} for a systematic classification of topographic attributes.

\begin{figure}[tb]
\begin{center}
    \includegraphics[trim=0 12 0 0,clip]{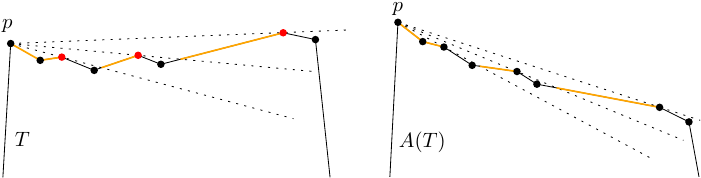}
  \end{center}

  \caption{Left: a TIN (in  $\mathbb{R}^2$)
%a TIN (1.5D profile) 
with three peaks and one viewpoint ($p$), with a viewshed composed of three parts (visible parts shown orange). Right: transformation of the terrain with no peaks (other than $p$) but the same viewshed complexity. Dotted segments show lines of sight from $p$.}
  \label{fig:flattening}
\end{figure}

One very simple and natural global measure of the ruggedness of a terrain that is relevant for viewshed complexity is to simply count the number of \emph{peaks} (i.e., local maxima) in the terrain.
It has been observed that areas with higher elevation difference, and hence, more peaks, cause irregularities in viewsheds \cite {franklin,klms-papgpt-14}, and this idea aligns with our theoretical understanding: the quadratic example from Fig.~\ref{fig:quadratic} is designed by creating an artificial row of peaks, and placing a viewpoint behind them.
However, while it seems reasonable to use the peak count as complexity measure, there is no theoretical correlation between the number of
peaks and the viewshed complexity.
This is easily seen by performing a
simple trick: any terrain can be made arbitrarily flat by scaling it
in the $z$-dimension by a very small factor, and then it can be rotated slightly.
%\maarten [suggests] {"and then it can be rotated slightly"? (the rotating is not to make it flat)}\maria{Done}
This results in a valid terrain without any peaks, but
retains the same viewshed complexity.
See Fig.~\ref{fig:flattening} for an example in $\mathbb{R}^2$.
%See Fig.~\ref{fig:flattening} for an example in 1.5D.
In fact, viewshed complexity is
invariant under affine transformations (i.e.,~scalings,
  rotations, and translations) of the terrain: the application of any affine combination to the terrain and the viewpoints results in a viewshed of the same complexity.
  Hence, any measure that has provable correlation with it must be affine-invariant as well.
\blue{This is a common problem to establish theoretical guarantees on viewshed complexity, or to design features of ``realistic'' terrains in general~\cite{BergHT10,MOET200848}.}
In fact, it is easy to see that none of the terrain attributes mentioned above is affine-invariant.

%(chnaged) This is a common problem when trying to establish theoretical guarantees on viewshed complexity, or indeed when designing features of ``realistic'' terrains in general~\cite{BergHT09,MOET200848}.

%\maarten {Since our introduction is quite long, I would consider using actual subsections rather than some of the subparagraphs for a full version. E.g.:} \maria{Sounds good}
%\subparagraph{Prickliness}
\subsection {Prickliness}

%In this work we explore a new topographic attribute:
In this work we propose a new topographic attribute:
%that is meant to describe the potential of a terrain to have high complexity viewsheds:
the \emph{prickliness}.
The definition follows directly from the above observations: it counts the number of peaks in a terrain, but does so {\em for every possible affine transformation} of the terrain.
%While the concept is independent of the type of terrain model used, for the sake of clarity,
We first present a definition for TINs, and then we explain how the definition carries over to DEMs.

%Formally, let $T$ be a polyhedral surface. We say that $T$ is a {\em terrain} if the surface is $xy$-monotone; that is, if any vertical line intersects $T$ at most once.
%We define the number of {\em peaks} or {\em local maxima} of $T$, $m(T)$, to be the number of internal vertices of $T$ which are extremal in the $z$-direction; that is, all adjacent vertices have a lower $z$-coordinate.\footnote{For technical reasons, to simplify our presentation we ignore vertices with neighbors at the same height.}

%Formally, let $T$ be a polyhedral surface. We say that $T$ is a {\em terrain} if the surface is $xy$-monotone; that is, if any vertical line intersects $T$ at most once.\maria{I would like to replace the previous two sentences with the following; would this create any subtle problem I don't see?
%Formally, 
Let $T$ be a %TIN, that is, a 
triangulated surface that is $xy$-monotone.
\new{A point $p$ on $T$ will be considered {\em concave} (resp., {\em convex}) if there exists a non-vertical plane through $p$ that leaves all neighboring points above it (resp., below it).
A vertex of $T$ will be called \emph{internal} if it is not on the boundary of the triangulation, i.e., if it is not incident to the unbounded face of the triangulation (considered in 2D).
}%\rodrigo{Better definitions of internal are welcome, I'm not too happy with this one}

% (i.e., any vertical line intersects it at most once).
Let $A$ be an affine transformation 
\new{from $\mathbb{R}^3$ to $\mathbb{R}^3$, defined as an invertible linear transformation followed by a translation (i.e., $A(x)=Mx+U$, for an invertible matrix $M \in \mathbb{R}^{3 \times 3}$ and $U \in \mathbb{R}^{3}$).}
Let $A(T)$ be the \new{polyhedral surface} obtained after applying $A$.
We define %the number of {\em peaks} %or {\em local maxima} 
%of $A(T)$, 
$m(A(T))$ to be the number of internal and convex vertices of $T$\footnote {We explicitly only count vertices that are already convex in the {\em original} terrain, since some affine transformations will transform local minima / concave vertices of the original terrain into local maxima.} that are local maxima in $A(T)$\footnote{See Section~\ref{sec:prelim} for a formal definition of local maxima}.
%; that is, all adjacent vertices have a lower or equal $z$-coordinate.
Let ${\cal A}(T)$ be the set of all affine transformations of $T$.

We define the {\em prickliness} of $T$, $\pi(T)$, to be the maximum number of local maxima over all transformations of $T$;
%\footnote{It might happen that, in the affine transformation achieving prickliness, some of the vertices considered local maxima have neighbors at the same height, which might be considered non-desirable. However, under certain reasonable non-degeneracy assumptions for the terrain, there exists a small perturbation of that transformation giving one fewer local maxima and such that in that transformation all vertices considered local maxima have all neighbors at strictly lower height. An assumption guaranteeing this property is that no two edges of $T$ have the same slope, and no two faces are parallel.} 
that is, $\pi(T) = \max_{A \in {\cal A(T)}} m(A(T))$.

%\rodrigo{Right here there is another comment  defining terrain-preserving affine transformations, which is directly related to a comment by Reviewer 4 (2nd comment about Page 4). Why is this commented?} \maria{I think it is commented because if it is not the definition we use. There are two options: (a) consider all affine transformations; (b) consider only those affine transformations that preserve the property of being a terrain. (b) is commented because we opted for (a).} 
%
%Let $A$ be an affine transformation. We say $A$ is {\em terrain-preserving} if $A(T)$ is terrain when $T$ is a terrain.
%Let ${\cal A}(T)$ be the set of all terrain-preserving affine transformations \rodrigo{of $T$?}.
%Then we define the {\em prickliness} of $T$, $\pi(T)$, to be the maximum number of local maxima over all transformation of $T$; that is, $\pi(T) = \max_{A \in {\cal A(T)}} m(A(T))$.
%
We start by observing that, essentially, the prickliness considers all possible {\em directions} in which the number of local maxima are counted.
%, and we can, in fact, provide an alternative definition that will be helpful in our analysis and computation of prickliness.
Let $\vec v$ be a vector in $\mathbb{R}^3$. Let $\pi_{\vec v} (T)$ be the number of internal and convex vertices of $T$ that are local maxima of $T$ in direction $\vec v$. 
%; that is, the number of internal and convex vertices of $T$ for which the local neighborhood does not extend further than that vertex in direction $\vec v$.
Then, $\pi(T) = \max_{\vec v} \pi_{\vec v} (T)$ (see See Section~\ref{sec:prelim} for a proof).

\begin{figure}[tb]
\centering
\includegraphics[scale=0.73]{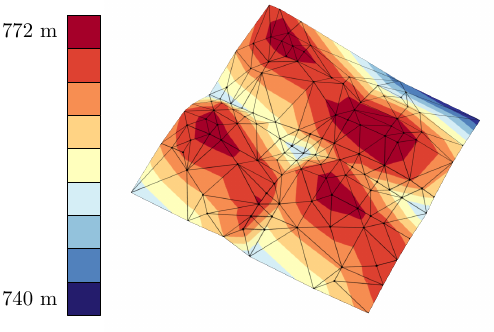}
\includegraphics[scale=0.73]{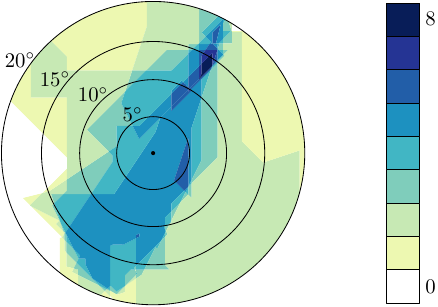}
\caption {(left) A TIN $T$, with triangulation edges shown in black, and elevation indicated using colors. 
(right) A visualization of the prickliness of $T$ as a function of the angles $(\theta, \phi)$ that define each direction (circles indicate contour lines for $\theta$); color indicates prickliness.
The maximum prickliness is $8$, attained at a direction of roughly $\theta=13^\circ$ and $\phi=60^\circ$ (north-east from the origin).}
\label {fig:example}
\end {figure}

  Using this observation, we reduce the space of all affine transformations to the 2-dimensional space of all directions in 3D.
  %Further note that, 
  Since $T$ is a terrain, for any $\vec v$ with a negative $z$-coordinate we have $\pi_{\vec v} (T) = 0$ by definition, thus the interesting directions reduce to the points on the (positive) unit half-sphere.
  This provides a natural way to visualize the prickliness of a terrain.
  Each direction can be expressed using two angles $\theta$ and $\phi$
  (i.e., using spherical coordinates), where $\theta$ represents the
  polar angle and $\phi$ the azimuthal angle. Fig.~\ref {fig:example}
  shows a small terrain and the resulting prickliness, showing a
  projection of the half-sphere, where each point represents a
  direction, and its color indicates its prickliness.\footnote {Note that we
  specifically define prickliness to be the \emph{maximum} over all
  orientations rather than, say, the average over all
  orientations. Even for a terrain with high-complexity viewsheds like the one in
  Fig.~\ref{fig:quadratic}, the average number of peaks would still be
  relatively small since there are many orientations with a small
  number of peaks. Hence, such a definition would be unlikely to accurately capture the
  complexity of viewsheds on a terrain. \new{Notice also that some orientations might result in an object that is no longer a terrain, but considering these orientations seems necessary because the associated peaks might still produce high complexity viewsheds.}}

%\maria{Specify somewhere that only convex vertices of the terrain are relevant in the definition?}
%\maarten {Yes, I think that's nice. Especially because it makes explicit that we don't allow shenanigans with flipping the terrain over. Maybe as another observation?}

%\begin {observation}
 % For any direction $\vec v$, only vertices that are {\em convex} in $T$ potentially contribute to $\pi_{\vec v} (T)$.
%\end {observation}

\paragraph{DEM terrains}
We note that all previous notions easily translate to DEMs. The
centers of the DEM cells can be seen as the {\em vertices} of the
terrain, and every internal vertex of the terrain has eight neighbors
given by the cell centers of the eight neighboring cells. Hence, in
the definitions for DEMs, the notion of {\em adjacent vertices} for
TINs is replaced by that of {\em neighbors}. %Then we have that
% A vertex is a local maximum for some affine transformation of the terrain if all of its neighbors have a lower $z$-coordinate. 
Analogously to the case of TINs, we can define local maxima based on the height of the neighbors of a vertex (see Section~\ref{sec:prelim} for a formal definition of local maximum).
This gives an equivalent definition of $\pi(T)$ 
when
$T$ is a DEM.
%\rodrigo{Now this looks a bit strange, because we have not given the definition of local maximum for TINs yet. Maybe this one should also wait, or we give an intuitive definition for TINs in the previous section} \maria{True. Is Section 2 only about TINs? In the affirmative, we could write here something emphasizing a bit more that we are going to give a definition only for DEM terrains such as "For DEM terrains, a vertex is considered a local maximum for some affine transformation..." (and at the beginning of Section 2 we should mention that the discussion only concerns TIN terrains)}
%
%\rodrigo{Looking at this again, I think that the issues that made us defer the definition of local maximum to the next section also exist for DEMs: a local maximum could be one single cell or a group of cells at the same height. So maybe it makes most sense to define maxima for DEMs after we properly defined it for TINs (i.e., after Section 2)}

%\maarten {For this interpretation of a DEM, though, our definition of "visibility" is not well-defined: what does it mean for a line segment to intersect the interior of the terrain when we model it as a non-planar graph?} \maria{I agree. Which definition we use in our experiments?}
%\rodrigo{I added the following note:}

In the case of DEMs, the definition of \emph{visibility} between two points needs to be adapted in order to compute viewsheds.
In our experimental work, we rely on the software routines provided by ArgGIS Pro, in particular on their \emph{Geodesic Viewshed} package.
In this package, the visibility between two points is decided as follows.
The line of sight between the two points is projected onto the spheroid representing the earth surface, resulting in a \emph{ground path}.
The ground path is sampled with \emph{step points} with consecutive distance proportional to the DEM cell-size. For each step point, it is checked if the terrain at that point obstructs the line of sight. The two points are considered visible if and only if no obstruction is found.
For more details on this we refer to ArcGIS Pro documentation~\cite{ArcGISViewshed}.

\paragraph{1.5D TIN terrains}
The most widespread terrain model in the computational geometry literature is the TIN (also called \emph{polyhedral terrain}). 
Visibility-related questions constitute an important family of problems concerning TIN terrains, but unfortunately some of these problems are quite difficult.
For this reason, terrains have also been defined and studied in one dimension less.
Standard TINs, as defined earlier in this paper, are \emph{2.5D TIN terrains}. 
If the dimension is reduced by one, we obtain \emph{1.5D terrains}, which can be seen as graphs of piece-wise linear univariate functions.
The simpler structure of 1.5D terrains makes it an interesting intermediate step towards the full understanding of problems in 2.5D TIN terrains.
% is that of {1.5D} TIN terrains (that is, TIN terrains in $\mathbb{R}^2$; the formal definition is given below). 
For this reason, 1.5D terrains and---in particular---visibility problems on them, have been thoroughly studied during the last 15 years. 
%Investigating prickliness of 1.5D terrains appears to have a theoretical interest as well, and for this reason in this paper we also address this question. 
In this work, first we study prickliness in 1.5D because it is conceptually easier, and then we investigate to what extent our results from the 1.5D case give insights into the 2.5D case. This is a common approach in the field.
%\maarten {What is the "theoretical interest" exactly? Maybe I would argue that we study the question in 1.5D first because it is conceptually easier, and then investigate to what extent our results from the 1.5D case give insight into the 2.5D case. (And say that this is a common approach in the field.)} \maria{Done}

More formally, a 1.5D TIN terrain is defined as an $x$-monotone polygonal line in $\mathbb{R}^2$. 
In this setting, a viewshed is composed of parts of terrain edges, and  the viewshed of one viewpoint can have a complexity that is linear on the  number of vertices of the terrain. Prickliness is defined as in the 2.5D case. Additionally, in this case it is also enough to consider all possible directions rather than all possible affine transformations (the proof of Observation~\ref{obs:pric} still applies unchanged). 
%\maarten {We could say "the proof of Observation 1 still applies unchanged".}\maria{Done} 
Notice that here directions are not vectors in $\mathbb{R}^3$ anymore, but vectors in $\mathbb{R}^2$. 

%From a theoretical standpoint, it will be also interesting to consider the simpler setting of {1.5D} terrains, 

\subsection{Results and organization} % or Contributions}
%\subparagraph{Results and organization} % or Contributions}

The remainder of this paper is organized as follows. 

We start with a theoretical block, composed of Sections~\ref{sec:compl-15D}-\ref{sec:comp25D}. This block is only concerned with TIN terrains because this is the most interesting model from a theoretical point of view and the one for which viewshed complexity is better understood---recall that there is not a unique way to measure viewshed complexity for DEMs. We  study two aspects of prickliness for TIN terrains, first in the easier case of $1.5$D terrains, and then for $2.5$D terrains. In Sections~\ref{sec:compl-15D}-\ref{sec:compl-25D}, we investigate the correlation between the prickliness of a terrain and the maximum complexity of the viewshed of a viewpoint.  We show that the prickliness of a $1.5$D terrain and the viewshed complexity of a single viewpoint are not related: we give examples where one is constant and the other is linear.
In contrast, unlike other measures of terrain ruggedness, there is a provable correlation between prickliness and viewshed complexity in $2.5$D.
In Sections~\ref{sec:comp15D}-\ref{sec:comp25D}, we investigate the computational problem of calculating the prickliness of a terrain.
We show that the prickliness of a 1.5D or 2.5D TIN terrain can be computed in polynomial time. The algorithm for the 1.5D case is optimal, while the one for the 2.5D case is near-optimal \new{under standard assumptions}. We also provide  an efficient approximate algorithm for 2.5D DEM terrains, which is used in our experiments.

In the second block of the paper, composed of Sections~\ref {sec:Topographic_attributes}-\ref{section:discussion}, we report on experiments that measure the values of distinct topographic attributes (including the prickliness) of real (2.5D) terrains, and analyze their possible correlation with viewshed complexity. From the experiments, we conclude that prickliness provides such a correlation in the case of TIN terrains, while the other measures perform more poorly. The situation for DEM terrains is less clear.

%\maarten {It is not very common in papers, but perhaps we could explicitly mark the two "blocks" using e.g. the part construct?} \maria{Done, see if you all like it}

\paragraph {Code}
Finally, we provide our code implementing two key algorithms for this
work: an algorithm to calculate the prickliness of a TIN terrain
(source code available from
\url{https://github.com/GTMeijer/Prickliness}; archived at
\href{https://archive.softwareheritage.org/swh:1:dir:c360f8c5b838bfe88910d26aad151dee69f69364;origin=https://github.com/GTMeijer/Prickliness;visit=swh:1:snp:f98118b47019c7c99d3f8ca36e76fdf56e9f5b39;anchor=swh:1:rev:7d1d9c1c9d76fee0d09beff2396006f8754e8613}{swh:1:dir:c360f8c5b838bfe88910d26aad151dee69f69364}),
and an algorithm to calculate the combined viewshed originating from a
set of multiple viewpoints (source code available from \url
{https://github.com/GTMeijer/TIN_Viewsheds}; archived at
\href{https://archive.softwareheritage.org/swh:1:dir:911b84528046c62ddd56c32905926748dd59791e;origin=https://github.com/GTMeijer/TIN_Viewsheds;visit=swh:1:snp:a61033758c7e06cfb67229206ddb2d2cfad7abd7;anchor=swh:1:rev:3228d926579e0631b255e7478f24963a508637b4}{swh:1:dir:911b84528046c62ddd56c32905926748dd59791e}).

\section {Preliminaries} \label{sec:prelim}

  We now review the precise definitions of the terms and concepts in this paper. 
  \new{Unless otherwise stated, all the concepts in this section apply to TINs and DEMs.}
  %\rodrigo{Frank, Maarten: please check the previous is correct}
  %\maarten {I would guess so - I always find it a bit confusing to argue about "neigbhours" in a grid, since they usually include diagonals and so the neighbour graph is non-planar. But thst doesn't seem to affect any of the arguments here.}
  
  %\maarten {Not sure if this should be a separate section or a subsection of the intro - since we do have some definitions before this?}
  
  Intuitively, a {\em local maximum} or {\em peak} in a terrain is a vertex that is higher than its neighbours. However, the definition is somewhat tricky due to the possibility of multiple adjacent vertices at the same height. We therefore recall the following definition \new{from~\cite {thesisRodrigo}},
  %\maria{Citation needed!} \maarten {I guess we could just cite Rodrigo's thesis. :) Or maybe he has a better source?} 
  which is common in the literature:
  
  \begin {definition}
    A {\em true local maximum} in a terrain is a vertex or connected group of vertices at the same height that are higher than all neighbors not in the group.
  \end {definition}
  
  While this definition correctly deals with the (degenerate) situation of multiple vertices at exactly the same height, it can be somewhat cumbersome to work with, and it makes it difficult to precisely define derivative concepts, such as the prickliness.
  
  Therefore, we consider the following alternative definition in this paper:
  
  \begin {definition}
    A {\em simple local maximum} in a terrain is a single vertex that is strictly higher than all its neighbors.
  \end {definition}

  %Note that, in generic (i.e., non-degenerate) terrains\rodrigo{Shouldn't we define what this means? I guess: not two neighboring vertices have the same height}\maarten {I added that it means non-degenerate, which is defined in the paragraph above. But we could consider adding a proper definition of these terms at the start of the section, since it is rather central to the whole discussion in this section.}, 
  Note that, \new{if no two neighboring vertices have the same height}, the two definitions are equivalent.
  However, in degenerate terrains \new{one can have a plateau of vertices all at the same height, which will never be a simple local maximum, but could be a true local maximum}.
  
  \begin {observation} \label {obs:tgs}
    In any given terrain, the number of true local maxima is greater than or equal to the number of simple local maxima.
  \end {observation}

  Now, let us consider the prickliness. Two definitions of local maximum give rise to two definitions of prickliness. 
  We will soon show that these are, in fact, equivalent (Lemma~\ref{lem:simpleistrue-1}), but until then,
  let us momentarily work with \emph{simple} local maxima, and thus let $\pi(T)$ denote the prickliness of $T$ using this definition of local maximum. Recall that $\pi_{\vec v} (T)$ denotes the number of internal and convex vertices of $T$ that are (simple) local maxima of $T$ in the direction of vector $\vec v$ in $\mathbb{R}^3$. 
  
  \begin {observation} \label{obs:pric}
  $\pi(T) = \max_{\vec v} \pi_{\vec v} (T)$.
 \end {observation}

\begin {proof}
  Clearly, for every vector $\vec v$ there exists an affine transformation $A$ such that $m(A(T)) = \pi_{\vec v}(T)$: take $A$ equal to the rotation that makes $\vec v$ vertical.
  We will show that also for every affine transformation $A$ there exists a vector $\vec v$ for which $m(A(T)) = \pi_{\vec v}(T)$.
  In particular, this then implies that the maximum value of $m(A(T))$ over all $A$ is equal to the maximum value of $\pi_{\vec v}(T)$ over all $\vec v$.

  Let $A$ be an affine transformation, and let $H$ be the horizontal plane $z=0$. Consider the transformed plane $H' = A^{-1}(H)$. Then any vertex of $T$ which has the property that all neighbors are on the same side of $H'$ in $T$, will be a local maximum or local minimum in $A(T)$.
  Now, choose for $\vec v$ the vector perpendicular to $H'$ and pointing in the direction which will correspond to local maxima.\qed
\end {proof}

It is easy to see that Observation~\ref{obs:pric} also holds for the variant of prickliness associated to true local maxima.
  
%  If we define the {\em true prickliness} as the maximum number of true local maxima over all directions, and the {\em simple prickliness} as the maximum number of simple local maxima over all directions, then we can state the following result.

If we denote the two variants of prickliness associated to the two definitions of local maximum as {\em true prickliness} and {\em simple prickliness}, we have the following:
  
  \begin {lemma} \label {lem:simpleistrue-1}
    For any terrain, the true prickliness is the same as the simple prickliness.
  \end {lemma}
  
  For the proof of Lemma~\ref {lem:simpleistrue-1}, let a direction $\vec v$ be {\em degenerate} for a given terrain $T$ is there are at least two vertices of $T$ that are in a common plane perpendicular to $\vec v$ (that is, that would have the same height after a transformation that makes $\vec v$ vertical).

%\maria{New version of the proof:}
\begin{proof}
Let $T$ be a terrain and let $\vec w$ be a direction. We denote by $t(\vec w)$ and $s(\vec w)$, respectively, the number of internal and convex vertices of $T$ that are true and simple local maxima of $T$ in the direction of vector $\vec w$. By Observation~\ref {obs:tgs}, we have $t(\vec w) \geq s(\vec w)$. Let $\vec v$ be a direction such that true prickliness is achieved at $\vec v$, i.e., the true prickliness of $T$ is equal to $t(\vec v)$. If $s(\vec v) = t(\vec v)$, by Observation~\ref {obs:tgs}, the simple prickliness is also achieved at $\vec v$ and we are done. Otherwise, by Observation~\ref {obs:tgs}, we have $t(\vec v) > s(\vec v)$. This means there must be one or several groups of adjacent vertices that are equally far in direction $\vec v$, so $\vec v$ is degenerate. Because the space of directions is continuous, there exists a sufficiently small $\eps$ such that there exists a perturbed direction $\vec v'$ with $|\vec v - \vec v'| < \eps$ that is not degenerate and such that $t(\vec v')=t(\vec v)$. Since $\vec v'$ is not degenerate,  $s(\vec v')=t(\vec v')$. This completes the proof.
\end{proof}

%  \begin {proof}
%    Let $T$ be a terrain, and let $\vec v$ be a direction such such that the number of true local maxima $t(\vec v)$ and the number of simple local maxima $s(\vec v)$ are different. Then by Observation~\ref {obs:tgs} we have $t(\vec v) > s(\vec v)$. This means there must be one or several groups of adjacent vertices that are equally far in direction $\vec v$, so $\vec v$ is degenerate. Because the space of directions is continuous, there exists a sufficiently small $\eps$ such that there exists a perturbed direction $v'$ with $|\vec v - \vec v'| < \eps$ that is not degenerate. Then we have $t(\vec v) = t(\vec v')$ and $s(\vec v) < s(\vec v')$. By Observation~\ref {obs:pric} the prickliness takes the maximum over all directions, so $\vec v$ will not define the \new{simple} prickliness; contrarily, the direction $\vec v^*$ that does define the \new{simple} prickliness must have $t(\vec v^*) = s(\vec v^*)$.
%  \end {proof}

%\maria{Somehow I find this proof not so clear. I would go: $\vec v$ is a direction giving true prickliness. If $t(\vec v) = s(\vec v)$, we are done. Otherwise, there exists $\vec v'$ such that $t(\vec v) = t(\vec v')$ and $s(\vec v') = t(\vec v')$. What do you think?} \maarten{That also seems fine - I don't have a strong preference.}

  That is, the two definitions of prickliness give the same result, and thus we will refer to them simply as {\em prickliness}. And since they are the same, but the simple local maxima definition is easier to work with, we will only use simple local maxima in the remainder of this paper, unless specifically stated otherwise.

\part{Theoretical results}

%\section {Prickliness and viewshed complexity: 1.5D}\label{subsec:compl-15D}
\section {Prickliness and viewshed complexity in 1.5D TIN terrains}\label{sec:compl-15D}
%\section {Prickliness and viewshed complexity} \label{sec:prick-view-compl}
 %\subsection{1.5D terrains}
 %Unfortunately, Conjecture~\ref{con:compl-prick} is not true for 1.5D terrains. In order to show it, we need to introduce some notation.

%\begin{figure}[tb]
 %\centering
	%\begin{minipage}{.48\textwidth}
	%\centering
	%	\vspace{27pt}
	%	\includegraphics[width=\linewidth]{angle1.pdf}
	%	\captionof{figure}{$se(v)$ for the corresponding vertices (shaded).}\label{fig:angle}
	%\end{minipage}
	%\hfill
	%\begin{minipage}{.48\textwidth}
  %\centering
	%	\includegraphics[width=\linewidth]{example15.pdf}
	%	\captionof{figure}{Terrains with low prickliness can have high viewshed complexity.}
	%	\label{fig:1.5Dex}
	%\end{minipage}
%\end{figure}

%The prickliness of a $1.5$D terrain and the viewshed complexity of a single viewpoint do not seem to be related. In order to show it, we need to introduce some notation.
\new{
In this section we show that the prickliness of a $1.5$D terrain and the viewshed complexity of a single viewpoint are not related, in the sense that there are examples where one is constant, and the other is linear. In order to show such an example, we need to introduce some notation.}

For every internal and convex vertex $v$ in $T$, we are interested in the vectors  $\vec w$ such that $v$ is a local maximum of $T$ in direction  $\vec w$. Note that \new{these vectors   $\vec w$ can be represented as unit vectors, and then the set of such vectors} becomes a region of the unit circle $\mathbb{S}^1$, which we denote by $se(v) \subset \mathbb{S}^1$. To find  $se(v)$, for each edge $e$ of $T$ incident to $v$ we consider the line $\ell$ through $v$ which is perpendicular to $e$. Then we take the \new{open} half-plane bounded by $\ell$ and opposite to $e$, and we translate it so that its boundary contains the origin. Finally, we intersect this half-plane with $\mathbb{S}^1$, which yields a half-circle. 
\new{The directions represented by this half-circle correspond to all the lines through $v$ that leave the interior of edge $e$ below.
By repeating this for the other edge incident to $v$, and intersecting the two half-circles associated to the two edges, we obtain a sector of $\mathbb{S}^1$ that corresponds to all lines through $v$ that leave both edges incident to $v$ below. These correspond to all directions in which $v$ is a local maximum.}
%For each direction $\vec w$ contained in the sector, the two corresponding edges do not extend further than $v$ in direction $\vec w$. Thus, $v$ is a local maximum in direction $\vec w$, 
It follows that this sector indeed represents $se(v)$, \new{and it can be computed in constant time for any $v$.} See Fig.~\ref{fig:angle} for an example. 

\begin{figure}[ht]
 \centering
 \includegraphics[scale=0.75]{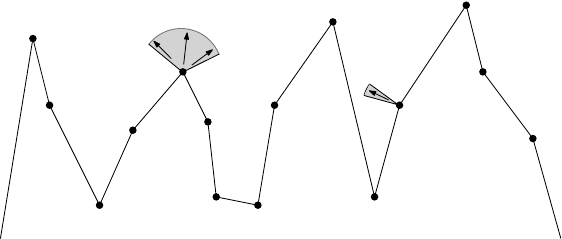}
 \caption{Example showing $se(v)$ (shaded) for two vertices.}\label{fig:angle}
\end{figure}

\begin {theorem}    \label{thm:1.5Dexample}
 There exists a 1.5D terrain $T$ with $n$ vertices and constant prickliness, and a viewpoint on $T$ with viewshed complexity $\Theta(n)$.
 \end {theorem}

\begin{proof}
 The construction is illustrated in Fig.~\ref{fig:1.5Dex-vase}, left. From a point $p$, we shoot $n/2$ rays in the fourth quadrant of $p$ such that the angle between any pair of consecutive rays is $2/n$. On the $i$th ray, there are two consecutive vertices of the terrain, namely, $v_i$ and $w_i$. The vertices are placed so that $\angle w_{i-1}v_iw_{i}=180-3/n$.
 
% \begin{figure}[ht]
%\centering
%\includegraphics[scale=0.9]{example15.pdf}
%		\captionof{figure}{Terrains with low prickliness can have high viewshed complexity.}
%		\label{fig:1.5Dex}
%\end {figure}

\begin{figure}[ht]
 \centering
	\includegraphics[scale=0.9]{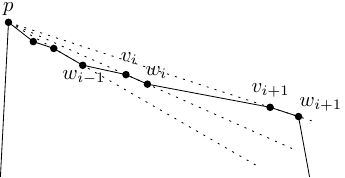} \quad \quad \quad
	\includegraphics[scale=0.65]{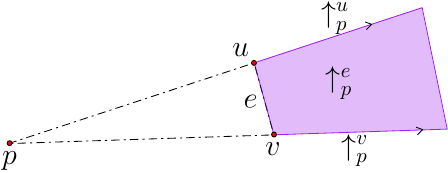}
		\caption{Left: terrains with low prickliness can have high viewshed complexity. Right: a vase.}
		\label{fig:1.5Dex-vase}
 \end{figure}

For every $i$, we have that $se(v_i)$ has angle $3/n$, while $se(w_i)$ is empty because $w_i$ is not convex. Since the angle between $w_{i-1}v_i$ and  $w_{i}v_{i+1}$ is $2/n$ and the angle between $v_iw_i$ and  $v_{i+1}w_{i+1}$ is also $2/n$, we have that $se(v_{i+1})$ can be obtained by rotating counterclockwise $se(v_i)$ by an angle of $2/n$. Thus, $se(v_i)\cap se(v_{i+1})$ has angle $1/n$, and $se(v_i)\cap se(v_{i+j})$ is empty for $j\geq 2$. We conclude that the prickliness of the terrain is constant.

If a viewpoint is placed close to $p$ along the edge emanating to the right of $p$, then for every $i$ the section $v_iw_iv_{i+1}$ contains a non-visible portion followed by a visible one. Hence, the complexity of the viewshed of the viewpoint is $\Theta(n)$. \qed
\end{proof}

% \section{Prickliness and viewshed complexity: 2.5D}
 \section{Prickliness and viewshed complexity in 2.5D TIN terrains}
\label{sec:compl-25D}

%\subsection{2.5D terrains}
%\begin{wrapfigure}[10]{r}{0.33\textwidth}
%  \begin{center}
   % \includegraphics[width=0.32\textwidth]{vase.pdf}
  %\end{center}
  %\caption{A vase.}
%   \vspace{2mm}
  %\label{fig:vases}
  %\end{wrapfigure}

Surprisingly, and in contrast to Theorem~\ref {thm:1.5Dexample},
we will show in Theorem~\ref{thm:2.5comp} that in 2.5D there is a provable relation between prickliness and viewshed complexity. 

We recall some terminology introduced in~\cite{fishy2014}.
%we do believe the conjecture is true.
%    
%We first consider the case where the terrain has only one viewpoint $p$. 
%
%
 Let $v$ be a vertex of $T$, and let $p$ be a viewpoint.  We denote by $\uparrow_p^v$ the half-line with origin at
$p$ in the direction of vector $\overrightarrow{pv}$. Now, let
$e=uv$ be an edge of $T$.  The \emph {vase} of
$p$ and $e$, denoted $\uparrow_p^e$, is the \new{unbounded region defined by}
%region bounded by 
$e$,
$\uparrow_p^u$, and $\uparrow_p^v$ (see Fig.~\ref{fig:1.5Dex-vase}, right). 

%\begin{figure}[!ht]
 %\centering
%\includegraphics[scale=0.6]{vase.pdf}
%\caption{A vase.}
%\label{fig:vases}
%\end{figure}

Vertices of the viewshed of $p$ can have three types~\cite{fishy2014}. A vertex of type 1 is a vertex of $T$, of which there are clearly only $n$.
A vertex of type 2 is the intersection of an edge of $T$ and a vase.
A vertex of type 3 is the intersection of a triangle of $T$ and two vases.
With the following two lemmas we will be able to prove Theorem~\ref{thm:2.5comp}.

\begin {lemma} \label{cl:type2}
There are at most $O(n \cdot \pi(T))$ vertices of type $2$.
\end {lemma}

\begin {proof}
Consider an edge $e$ of $T$ and let $H$ be the plane spanned by $e$ and $p$. Consider the viewshed of $p$ on $e$. Let $qr$ be a maximal invisible portion of $e$ surrounded by two visible ones. Since $q$ and $r$ are vertices of type 2, the open segments $pq$ and $pr$ pass through a point of $T$. 
On the other hand, for any point $x$ in the open segment $qr$, there exist points of $T$ above the segment $px$. This implies that there is a continuous portion of $T$ above $H$ such that the vertical projection onto $H$ of this portion lies on the triangle $pqr$. Such portion has a \new{true} local maximum in the direction perpendicular to $H$ which is a convex and internal vertex of $T$.
%\maria{TO DO: Argue why the local maxima is not a boundary vertex of the terrain}\maarten{There is a triangle formed by p and the two endpoints of e, and the local maximum we find must lie in the interior of this triangle, and therefore in the interior of the terrain, no?} 
In consequence, each invisible portion of $e$ surrounded by two visible ones can be assigned to a distinct point of $T$ that is a local maximum in the direction perpendicular to $H$. Hence, in the viewshed of $p$, $e$ is partitioned into at most $2 \pi(T) +3$ parts.\footnote{We obtain $2 \pi(T) +3$ parts when the first and last portion of $e$ are invisible; otherwise, we obtain fewer parts.} \qed
%Let $H$ be the plane spanned by $e$ and $p$. 
%Suppose $e$ is partitioned into $k$ visible and invisible parts. Then the invisible parts must be invisible because of a piece of $T$ sticking out above $H$. But these pieces are surrounded by empty space in $H$ (otherwise the visible pieces would not be visible), therefore, in the direction perpendicular to $H$, these are local maxima.
\end{proof}

\begin{figure}[tb]
  \centering
   \includegraphics[scale=0.6]{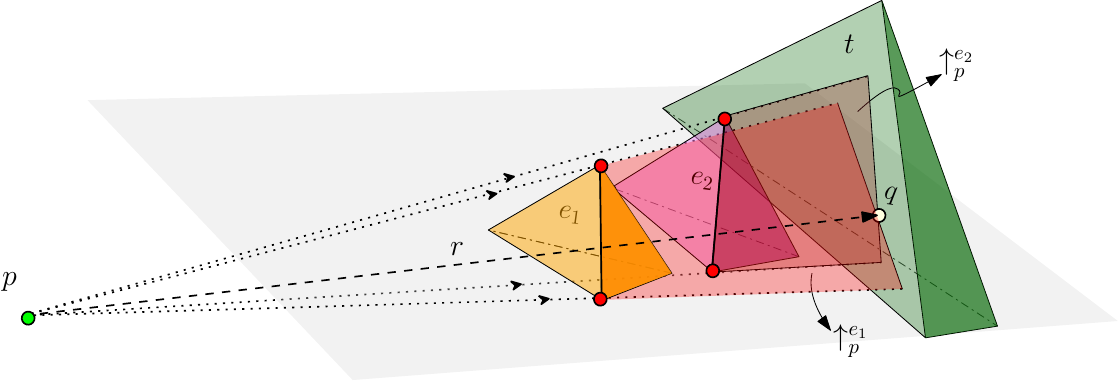} 
  \caption{The situation in the proof of Lemma~\ref{cl:type3}.}

  \label{fig:type3}
 \end{figure}

    \begin {lemma} \label{cl:type3}
There are at most $O(n \cdot \pi(T))$ vertices of type $3$.
%The number of vertices of type $3$ is at most $n$ plus the number of vertices of type $2$.
      %There are at most as many vertices of type $3$ as vertices of type $2$.
    \end {lemma}

    \begin {proof}
Let $q$ be a vertex of type 3 in the viewshed of $p$. Point $q$ is the intersection between a triangle $t$ of $T$ and two vases, say,  $\uparrow_p^{e_1}$ and $\uparrow_p^{e_2}$; see Fig.~\ref{fig:type3}. Let $r$ be the ray with origin at $p$ and passing through $q$. Ray $r$ intersects edges $e_1$ and $e_2$. First, we suppose that $e_1$ and $e_2$ do not share any vertex and, without loss of generality, we assume that $r\cap e_1$ is closer to $p$ than $r\cap e_2$. Notice that $r\cap e_2$ is a vertex of type 2 because it is the intersection of $e_2$ and $\uparrow_p^{e_1}$, and $\uparrow_p^{e_1}$ partitions $e_2$ into a visible and an invisible portion. Thus, we charge $q$ to $r\cap e_2$. If another vertex of type 3 was charged to $r\cap e_2$, then such a vertex would also lie on $r$. However, no point on $r$ after $q$ is visible from $p$ because the visibility is blocked by $t$. Hence, no other vertex of type 3 is charged to $r\cap e_2$.

If $e_1$ and $e_2$ are both incident to a vertex $v$, since $t\, \cap \uparrow_p^{e_1} \cap \uparrow_p^{e_2}$ is a type 3 vertex, we have that $r$ passes through $v$. Therefore, $q$ is the first intersection point between $r$ (which can be seen as the ray with origin at $p$ and passing through $v$) and the interior of some triangle in $T$. Therefore, any vertex $v$ of $T$ creates at most a unique vertex of type 3 in this way.\qed
      %We charge a vertex of type 3 to the vertex of type 2 that projects to it from $p$.
    \end {proof}

  \begin {theorem}
  \label{thm:2.5comp}
      The complexity of a viewshed in a 2.5D terrain is $O(n \cdot \pi(T))$. 
    \end {theorem}

Next we describe a construction showing that the theorem is best possible.

\begin{theorem}
\label{obs:lowerbound2.5D}
There exists a 2.5D terrain $T$ with $n$ vertices and prickliness $\pi(T)$, and a viewpoint on $T$ with viewshed complexity $\Theta(n \cdot \pi(T))$.
\end{theorem}

%The previous observation does not require the construction of Theorem~\ref{thm:1.5Dexample}. \maria{What does this sentence mean?}
%\maarten [guesses] {Maybe it means that this lower bound follows from a  new construction, and not an adaptation/generalization of the one in Theorem~\ref{thm:1.5Dexample}?} \maria{Probably\ldots Do we want to rephrase this?} \maarten {We can also just leave it out.}

\begin{proof}
Consider the standard quadratic viewshed construction, composed of a set of front mountains and back triangles (Fig.~\ref{fig:lower_bound} (left)).
Notice that there can be at most $\pi(T)$ mountains ``at the front''.
We add a surrounding box around the construction, see Fig.~\ref{fig:lower_bound} (right), such that each vertex of the back triangles is connected to at least one vertex on this box.
We set the elevation of the box so that it is higher than all the vertices of the back triangles, but lower than those of the front mountains.
%\maria{Lower than those of the bounding box is because we don't want to mess up with the mointains at the front?}\maarten{I think it was to prevent them from being convex (so they will never be peaks in any direction, which might mess up the prickliness).}
In this way, no vertex of the back triangles will be a local maximum in any direction, and all local maxima will come from the front. \qed%\maria{We should specify where is the viewpoint} 
\end{proof}

\begin{figure}[t] 
\centering\includegraphics[scale=0.5]{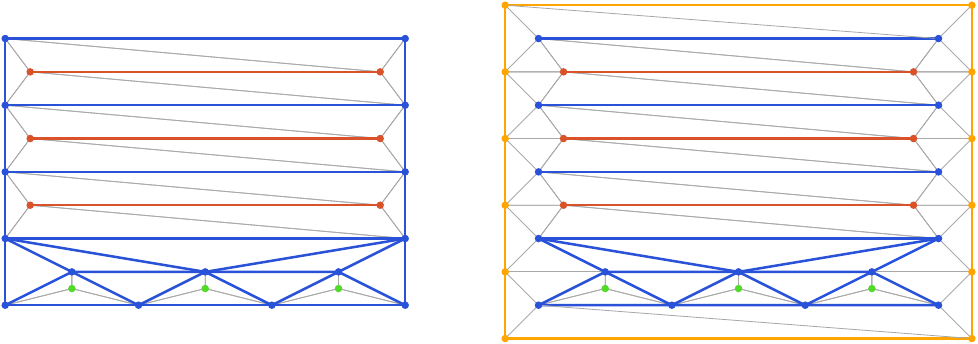}
\caption{
\new{Left: Schematic top-down view of the classic quadratic construction; see Fig.~\ref {fig:quadratic} for a  perspective view with quadratic complexity.
Right: Construction adapted to have small prickliness.
Blue vertices/edges are low, red are medium hight, and green are high.}
The construction on the right introduces a new height (yellow) between medium and high, and changes the triangulation slightly, to ensure that all convex vertices in the construction are green.}
%\maria{The right figure is the view from top? We should say it}\maarten {Yes. Also, I don't think we need the left picture, since it is in the introduction and we can just refer to it (and that one is nicer anyway). But maybe we want to have a top-down view of the original construction, and one of the adapted construction?}}
%\maarten[changed the figure and caption]{}
\label{fig:lower_bound}
\end{figure}

\section {Prickliness computation in 1.5D TIN terrains}
\label {sec:comp15D}

%  In order to verify Conjecture~\ref {con:rea}, it is useful to be able to calculate the prickliness of a terrain.
  
%    \maarten {I would say that's not the main reason why we want to be able to calculate prickliness...}\maria{Indeed. Any alternative inspirational sentence, or shall we just write nothing?}
%    \maarten {I think we can just leave it out.}
    
 \subsection{Algorithm} \label {sec:alg1}
% \subsubsection{Algorithm} \label {sec:alg1}

%For when the definition of $se(v)$ is in the Appendix:
%For every internal and convex vertex $v$ in $T$, we are interested in the vectors  $\vec w$ such that $v$ is a local maximum of $T$ in direction  $\vec w$. These feasible vectors  $\vec w$ can be represented as unit vectors, and then the feasible set becomes a region of the unit circle $\mathbb{S}^1$, which we denote by $se(v)$. More precisely, $se(v)$ is a sector and can be computed in constant time (see Appendix~\ref{app:sev}). We sometimes write  $se(v)=[\alpha,\beta]$, where $\alpha$ and $\beta$ are the angles bounding the sector.

For every internal and convex vertex $v$ in $T$, we compute \new{the circular sector} $se(v)$ in constant time, \new{by combining the directions perpendicular to the two edges incident to $v$}, as explained in Section~\ref{sec:compl-15D}.
The prickliness of $T$ is the maximum number of sectors of type $se(v)$ whose intersection is non-empty. 
\new{To compute it,  first we sort the bounding angles of the sectors, distinguishing when a sector begins and when one ends. Then we go through them, in order, keeping track of how many sectors contain each sub-sector between two consecutive angles. 
The maximum number found is the prickliness.
This can be done  in $O(n \log n)$ time.
}
Thus, we obtain:

\begin{theorem}
 The prickliness of a 1.5D terrain can be computed in $O(n\log n)$ time. 
\end{theorem}

%\begin{theorem}
%In 1.5D, the prickliness can be computed in $O(n\log n)$ time. 
%\end{theorem}

\subsection{Lower bound} \label{subsec:red-15}
%\subsubsection{Lower bound} 

Now we show that $\Omega(n \log n)$ is also a lower bound for computing the prickliness of a 1.5D terrain. 
The reduction is from the problem of checking distinctness of $n$ integer elements, which has an $\Omega(n \log n)$ lower bound in the bounded-degree algebraic decision tree model\footnote{\new{This is a well-known computational model to solve decision problems involving numbers, bounding  how many comparisons are needed to arrive to a solution.}}~\cite{LubiwR-lb,Yao-lb}.%~\cite{LubiwR-lb}.

Suppose we are given a set ${\cal S}=\{x_1,x_2,\ldots,x_n\}$ of $n$ integer elements, assumed without loss of generality to be positive. We multiply all elements of $\cal S$ by $180/(\max{\cal S}+1)$ and obtain a new set ${\cal S'}=\{x'_1,x'_2,\ldots,x'_n\}$ such that $0 < x'_i < 180$, for each $x'_i$. 
We construct a terrain $T$ \new{such that computing its prickliness allows to determine if all elements of ${\cal S}$ are unique}. 

\new{For each $x'_i$, the goal is to create in $T$ a convex vertex $v_i$ such that $se(v_i)=(x'_i - \varepsilon, x'_i + \varepsilon)$, where $\varepsilon =18/(\max{\cal S}+1)$, and such that its two neighbors are at distance $1$ from $v_i$.\footnote{We sometimes write  \new{$se(v)=(\alpha,\beta)$}, where $\alpha$ and $\beta$ are the angles bounding the sector.} This is possible because, as explained in Section~\ref{sec:compl-15D},  $se(v_i)$ is determined by the slopes of the two edges incident to $v_i$, so given $se(v_i)$ we can infer the corresponding slopes. Together with the fact that they are at distance $1$ from $v_i$, this completely determines the positions of the  vertices incident to $v_i$ to its left and right. These vertices are denoted  by $w^l_i$ and $w^r_i$, respectively. See Fig.~\ref{fig:lb} for an example.}
%For each $x'_i$, we create in $T$ a convex vertex $v_i$ such that $se(v_i)=[x'_i - \varepsilon, x'_i + \varepsilon]$, where $\varepsilon =18/(\max{\cal S}+1)$, and such that its two neighbors are at distance $1$ from $v_i$.\footnote{We sometimes write  $se(v)=[\alpha,\beta]$, where $\alpha$ and $\beta$ are the angles bounding the sector.} See Fig.~\ref{fig:lb} for an example. We denote the incident vertices to $v_i$ to its left and right by $w^l_i$ and $w^r_i$, respectively.

\begin{figure}[t]
 \centering
 \includegraphics[scale=0.4]{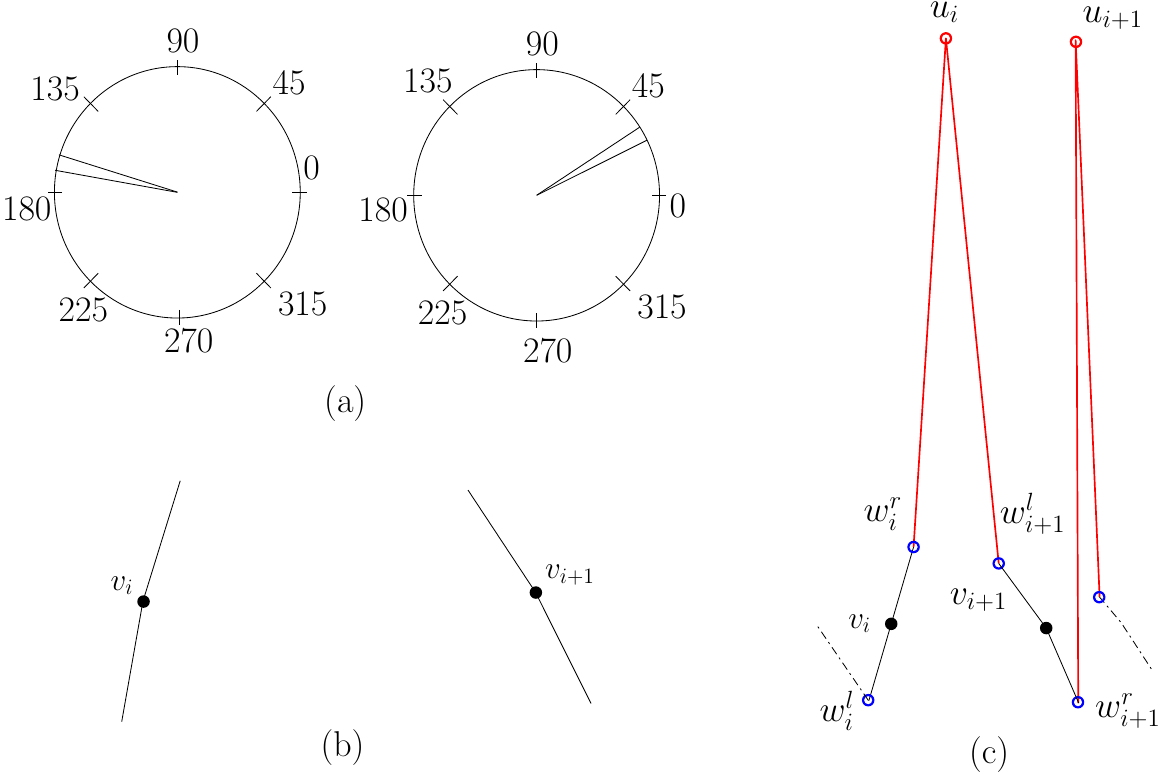}
 \caption{(a) Sectors associated to the element set $\{160,25\}$. (b) The corresponding convex vertices. (c) Construction of the terrain.}\label{fig:lb}
 \end{figure}

We arrange these convex vertices in the order of the elements in ${\cal S'}$ from left to right, and we place all of them at the same height.
Then we place a dummy vertex $u_i$ between every pair of consecutive vertices $v_i$ and $v_{i+1}$, and connect $u_i$ to $w^r_i$ and $w^l_{i+1}$; see Fig.~\ref{fig:lb}c. The height of $u_i$ is chosen so that its two neighbors become concave vertices, and also so that $[\min{(\cal S')}- \varepsilon,\max{(\cal S')}+ \varepsilon] \subseteq se(u_i)$: \new{In detail, the slope of $w^r_iu_i$ should be greater than $\max\{\mathrm{slope\ of\ } v_iw^r_i,0,\max{(\cal S')}+ \varepsilon-90º\}$, which guarantees that $w^r_iu_i$ has positive slope, $w^r_i$ is concave and $se(u_i)=(\alpha,\beta)$, with $\beta\geq \max{(\cal S')}+ \varepsilon$. Additionally, the slope of $u_iw^l_{i+1}$ should be smaller that $\min\{\mathrm{slope\ of\ } w^l_{i+1}v_{i+1},0,\min{(\cal S')}- \varepsilon -90º\}$, which guarantees that $u_iw^l_{i+1}$ has negative slope, $w^l_{i+1}$ is concave and $\alpha \leq \min{(\cal S')}- \varepsilon$.}
%This is possible because moving $u_i$ upwards increases its feasible region $se(u_i)$, the limit being $[0, 180]$.
The following lemma allows us to prove Theorem~\ref{thm:lower} below.

\begin{lemma}\label{lem:lb}
The prickliness of $T$ is $n$ if and only if all elements in ${\cal S}$ are distinct.
\end{lemma}

\begin{proof}
Every vertex of type $u_i$ satisfies that $ [\min{(\cal S')} - \varepsilon,\max{(\cal S')}+ \varepsilon] \subseteq se(u_i)$. Thus the prickliness of $T$ is at least $n-1$. For every vertex of type $v_i$, \new{$se(v_i)=(x'_i - \varepsilon, x'_i + \varepsilon)$}. Finally, the vertices of type $w^l_i$ and $w^r_i$ are concave (or not internal, in the case of $w^l_1$ and $w^r_n$) so $se(w^l_i)$ and $se(w^r_i)$ are empty.

Consequently, \new{
the prickliness of $T$ will be $n-1+M$,
where $M$ is the maximum number of occurrences of an element of ${\cal S}$.
It follows that the prickliness of $T$ is $n$ if and only 
%the sectors of type $se(v_i)$ are pairwise disjoint, which happens if and only if 
the elements in ${\cal S}$ are all distinct.}\qed
\end{proof}

\new{Since $T$ can be constructed in $O(n)$ time, we conclude the following:}

\begin{theorem}
\label{thm:lower}
The problem of computing the prickliness of a 1.5D terrain has an $\Omega(n \log n)$ lower bound in the bounded-degree algebraic decision tree model.
\end{theorem}

\section{Prickliness computation in 2.5D terrains}
\label {sec:comp25D}

In this section, we consider the problem of computing the prickliness of a 2.5D terrain.
In Section~\ref {sub:Algorithm_for_TINs}, we present a simple quadratic-time algorithm for TIN terrains. 
In Section~\ref {sub:Algorithm_for_DEMs}, we discuss how to adapt it to DEM terrains; such an adaptation is needed to be able to run experiments for DEM terrains.
Finally, in Section~\ref  {sub:Lower_bound_2.5D}, we go back to TIN terrains, and provide evidence that the complexity of the algorithm in Section~\ref {sub:Algorithm_for_TINs} is close to the best possible.

\subsection{Algorithm for TINs} \label {sub:Algorithm_for_TINs}
%\subsubsection{Algorithm}

%A basic observation is that the sphere of potential directions is subdivided into $n^2$ cells by the $n$ planes through the origin parallel to the triangles of $T$. For any pair of directions $\vec v$ and $\vec w$ in the same cell, the values of $\pi_{\vec v}$ and $\pi_{\vec w}$ are equal. This gives a trivial $O(n^3)$ time algorithm to calculate the prickliness: compute this subdivision, for each cell take a single vector, and for this vector count the number of local maxima (testing whether a vertex is a local maximum is a local operation that takes time proportional to the degree of the vertex; the sum of degrees of all vertices is linear).
%\maria{This subdivision is different than the one we obtain in the algorithm below (if we were in 1.5D, it would roughly be a rotation of $\pi/2$ of the other, but in 2.5D I am not sure); but since the subdivisions are different, I don't think it can be that ``For any pair of directions $\vec v$ and $\vec w$ in the same cell, the values of $\pi_{\vec v}$ and $\pi_{\vec w}$ are equal", because this is the property we claim that holds for the other one}
%\maria{Do we want this introductory paragraph, or shall we go straight to the faster algorithm?}
%\maarten {I think it is helpful, but I'd say this is one of the first things to cut for space.}
%
%We propose a faster algorithm by extending the idea of 1.5D terrains to 2.5D terrains as follows: 

%We propose a faster algorithm that extends the idea from Section~\ref {sec:alg1} to 2.5D terrains as follows:
We propose an algorithm that extends the idea from Section~\ref {sec:alg1} to \new{a 2.5D terrain $T$} as follows:
For every convex terrain vertex $v$, we compute the region of the unit sphere $\mathbb{S}^2$ containing all vectors  $\vec w$ such that $v$ is a local maximum of $T$ in direction  $\vec w$. As we will see, such a region is a cone, which we denote by $co(v)$. 
Furthermore, we denote the portion of $co(v)$ on the surface of $\mathbb{S}^2$ by $co_{\mathbb{S}^2}(v)$.

In order to compute $co(v)$, we consider all edges of $T$ incident to $v$. Let $e=vu$ be such an edge, and consider the plane orthogonal to $e$ through $v$. Let $H$ be the \new{open} half-space which is bounded by this plane and does not contain $u$. We translate $H$ so that the plane bounding it contains the origin; let $H_e$ be the intersection of the obtained half-space with the unit sphere $\mathbb{S}^2$. The following property is satisfied: For any unit vector $\vec w$ in $H_e$, the edge $e$ does not extend further than $v$ in direction $\vec w$. We repeat this procedure for all edges incident to $v$, and consider the intersection $co(v)$ of all the obtained half-spheres $H_e$. For any unit vector $\vec w$ in $co(v)$, none of the edges incident to $v$ extends further than $v$ in direction $\vec w$. Since $v$ is convex, this implies that $v$ is a local maximum in direction $\vec w$.

Once we know all regions of type $co(v)$, computing the prickliness of $T$ reduces to finding a unit vector that lies in the maximum number of such regions. To simplify, rather than considering these cones on the sphere, we extend them until they intersect the boundary of a unit cube $\mathbb{Q}$ centered at the origin. The conic regions of type $co(v)$ intersect the faces of $\mathbb{Q}$ forming (overlapping) convex regions. Notice that the problem of finding a unit vector that lies in the maximum number of regions of type $co(v)$ on $\mathbb{S}^2$ is equivalent to the problem of finding a point on the surface of  $\mathbb{Q}$ that lies in the maximum number of ``extended" regions of type $co(v)$. \new{The second problem can be solved by computing, for each face of the cube, the maximum overlap of convex regions (by constructing the arrangement induced by the convex regions and traversing its dual).}
%The second problem can be solved by computing the maximum overlap of convex regions using a topological sweep~\cite{edels1989}, for each face of the cube.

 \new{For every convex vertex $v$, computing the intersection between the extended region $co(v)$ and the boundary of $\mathbb{Q}$ takes $O(\deg(v) \log(\deg(v)))$ time, where $\deg(v)$ is the degree of $v$ in $T$. Therefore, performing this operation for all convex vertices $v$ of $T$ takes time $O(\sum_v \deg(v) \log(\deg(v))) = O(2 |E| \log n) = O(n \log n)$ (where $|E|$ denotes the total number of edges of $T$). On the other hand, constructing the arrangement induced by the convex regions in each face of the cube and traversing its dual takes $O(n^2)$ time. We obtain the following:}

%Computing the intersection between the extended regions of type $co(v)$ (for all convex vertices $v$) and the boundary of $\mathbb{Q}$ takes $O(n \log n)$ time, and topological sweep to find the  maximum overlap takes $O(n^2)$ time. We obtain the following:

\begin{theorem} \label {thm:2D}
 The prickliness of a 2.5D terrain can be computed in $O(n^2)$ time. 
\end{theorem}

%\begin{theorem} \label {thm:2D}
%In 2.5D, the prickliness can be computed in $O(n^2)$ time. 
%\end{theorem}

\subsection{Algorithm for DEMs}
\label{sub:Algorithm_for_DEMs}

The prickliness of a DEM terrain can be computed using the same
algorithm as for TINs: Each cell center can be seen as a vertex $v$, and its neighbors are the cell centers of its eight
neighboring cells. The edges connecting $v$ to its neighbors can then
be used to compute $co(v)$ as in Section~\ref{sub:Algorithm_for_TINs}, and the
rest of the algorithm follows. However, DEM terrains have significantly more vertices, and vertices
have on average more neighbors; this causes a significant increase in
computation time and, more importantly, in memory usage. For this
reason, in our experiments, the prickliness values for the DEM
terrains were approximated.

The approximated algorithm discretizes the set of vectors that are
candidates to achieve prickliness as follows: For every interior cell
$g$ of the terrain, we translate a horizontal grid $G$ of size $n$ by
$n$ and cell size $s$ above the cell center $v$ of $g$ in a way that
$v$ and the center of $G$ are vertically aligned and at distance
one. Then the vectors considered as potential prickliness are those
with origin at $v$ and endpoint at some cell center $c$ in $G$. For
any such vector lying inside $co(v)$, the value of $c$ gets
incremented by one. When all interior cells of the terrain have been
processed, a cell center of $G$ with maximum value gives the
approximated prickliness of the terrain. Cell size $s$ was set to
$0.05$, based on the spread of the results on TIN terrains. This
method should, in practice, produce a close approximation of
prickliness.

%This algorithm is the last contribution of the theoretical block of this paper. We next move to the experimental block.

\subsection{Lower bound} \label {sub:Lower_bound_2.5D}

%\maarten {Thinking about it more, I think the following may be a bit cleaner (and as a bonus, result in a more believable terrain than the "pockets with small mountains in them" of the previous section).}

In this section we show that the problem of computing the prickliness of a 2.5D TIN terrain is 3SUM-hard.
This implies that our result in Theorem~\ref {thm:2D} is likely to be close to optimal: 
The best-known algorithm for \textsc{3sum} runs in $O(n^2(\log\log n)^{O(1)}/(\log n)^{2})$ time, and it is believed there are no significantly faster solutions \cite{DBLP:journals/talg/Chan20}.

We will reduce from the problem of covering a square by strips (defined below as \textsc{strips-cover-box}). In a nutshell, the idea is to map the square to a set of possible directions, and construct a terrain that has two vertices for each strip, such that one vertex is \new{a local maximum} for all directions on one side of the strip and the other vertex is \new{a local maximum} for all directions on the other side of the strip. Then,  the strips cover the square if and only if there is no direction such that, for every strip, one of the two associated vertices is \new{a local maximum} in that direction; that is, if and only if the prickliness is smaller than the number of strips.

\subsubsection {The \textsc{strips-cover-box} problem.}

Our reduction is from the problem \textsc{strips-cover-box}. In this problem, we are given a square $Q$ \new{in $\mathbb{R}^2$} and $n$  strips\footnote{A {\em strip} is the area between two parallel lines.} of infinite length. The problem is to answer the question: \new{Is every point in $Q$ contained in at least one strip?}
See Fig.~\ref {fig:strips-cover-box}a for an example.
Gajentaan and Overmars show that this problem is 3SUM-hard~\cite{gajentaan1995class}. 
%Since it is now known that 3SUM can be solved in $o(n^2)$ time\rodrigo{The previous is confusing: Indeed, 3SUM \emph{can} be solved in $o(n^2)$ time, but I think here the intention was to say that it is conjectured that no $O(n^{2-\varepsilon})$ algorithm  for 3SUM exists?} , 
\new{For completeness, we start by recalling some definitions and the precise statement of Gajentaan and Overmars.}

\begin{figure}[tb]
\begin{center}
    \includegraphics[page=1]{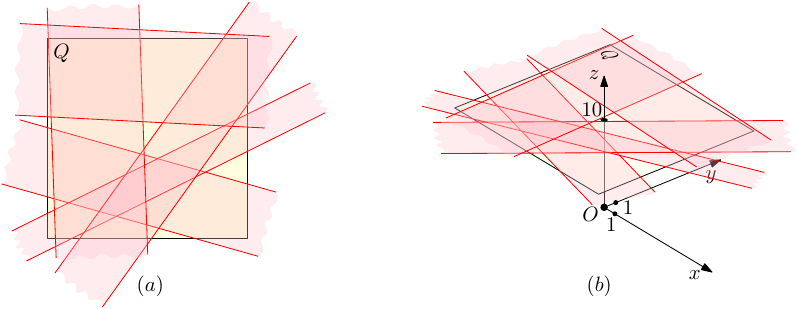}
  \end{center}
  \caption
  { (a) An instance of \textsc {strips-cover-box} (where the answer is NO).
    (b) Positioning the instance in space.
  }
  \label{fig:strips-cover-box}
\end{figure}

\begin {definition} [Definition 2.1 from~\cite {gajentaan1995class}]
  Given two problems \textsc{pr1} and \textsc{pr2}, we say that \textsc{pr1} is $f(n)$-solvable using \textsc{pr2} if and only if every instance of \textsc{pr1} of size $n$ can be solved using a constant number of instances of \textsc{pr2} of at most linear size and $O(f(n))$ additional time. 
  We denote this by
  \[
    \textsc{PR1} \lll_{f(n)} \textsc{PR2}.
  \]
\end {definition}

\begin {theorem} [combining Theorems 3.1, 3.2, and 6.1 from~\cite {gajentaan1995class}] \label {thm:strips-cover-box}
  \[
    \textsc{3SUM} \lll_{n \log n} \textsc{strips-cover-box}
  \]
\end {theorem}

Essentially, this implies that \textsc{strips-cover-box} is as hard as \textsc{3sum}, unless \textsc{3sum} can be solved faster than $\Omega(n \log n)$ time (which is considered highly unlikely~\cite{DBLP:journals/talg/Chan20}).

In the remainder of this section, we will set out to show that $\textsc{strips-cover-box} \lll_{n\log n} \textsc{prickliness}$.

\subsubsection {Mapping strips to directions.}

Now, assume we are given an instance of \textsc{strips-cover-box} of size $n$.
We will assume w.l.o.g. that $Q = [-1,1] \times [-1,1]$ is a
%unit\frank{remove unit?} 
$2$ by $2$ square centered at the origin.
We will describe how to construct from the instance a set of terrain features that we call {\em curtains}.

\begin{figure}[tb]
\begin{center}
    \includegraphics[page=3]{3sum_hard_new_proof.pdf}
  \end{center}
  \caption
  { (a) An instance of \textsc {strips-cover-box}.
    (b) Strips are replaced by pairs of half-planes.
    (c)~Half-planes are colored based on \new{their orientation (but only if they contain at least one corner of~$Q$)}. %\maarten {Not exactly, since a half-plane may contain multiple corners. It would be correct to say they are colored based on which corner of $Q$ is contained in it, and furtherst to its interior. Or equivalently, we can say they are colored based on their orientation, but only if they contain at least one corner of $Q$. Not sure which is the easier to read...} 
    Half-planes that do not intersect $Q$ are ignored (gray).
  }
  \label{fig:half-planes}
\end{figure}

The complement of each
strip is given by two half-planes. 
\new{
Consider the $2n$ half-planes that are the complements of all the given strips, together with the square
$Q$.}
See Fig.~\ref {fig:half-planes}b for an example.
Answering the above question is equivalent to answering whether
there exists a point in $Q$ covered by $n$ of the half-planes: \new{a point
inside a strip lies in none of the two half-planes corresponding to the
strip, whereas a point outside the strip lies in exactly one of those
two half-planes. Hence, a point lies outside all $n$ strips if and only
if it lies in exactly $n$ half-planes.}
To restrict our attention to the points in $Q$, we simply create four additional half-planes aligned with the sides of $Q$ (and containing $Q$). Then a point lies in $Q$ if and only if it is contained in all four half-planes.
Let $\cal H$ be the resulting set of $2n + 4$ half-planes.

\begin{observation} \label{obs1:3sum}
Any point is covered by at most $n+4$ half-planes in $\cal H$. If a point is covered by exactly $n+4$ half-planes in $\cal H$, then it lies inside $Q$. Such a point exists if and only if the answer to the \textsc{strips-cover-box} instance is NO.
\end{observation}

Now, we lift the plane containing $Q$ \new{and all the lines bounding the strips to} the horizontal plane in $\mathbb{R}^3$ at $z=10$, see Fig.~\ref {fig:strips-cover-box}b.
We identify each point $(x,y)$ in $Q$ with the direction vector $\overrightarrow {(x,y,10)}$.

\subsubsection {Constructing a set of curtains.}

We associate each half-plane \new{in $\cal H$} with a {\em curtain}. A curtain $(p, \vec d)$ is defined by a point $p$ in $\mathbb R^3$ together with a direction $\vec d$, and consists of the ray from $p$ in direction $\vec d$ together with all points vertically below this ray. 
\begin{figure}[tb]
\begin{center}
    \includegraphics[page=2]{3sum_hard_new_proof.pdf}
  \end{center}
  \caption
  { (a) A single \new{half-plane} bounded by a line $\ell$, lifted to the plane $z=10$, and a \new{point $d'$} inside the \new{half-plane}.
    (b) The plane $A$ through $\ell$ and the origin $O$ and its normal vector $\vec d$.
    (c) The curtain below the halfline from the origin in direction $\vec d$. The vertex of the curtain (at the origin) is a local maximum in direction $\vec d'$.
  }
  \label{fig:curtains-0}
\end{figure}
Each half-plane $H \in \cal H$ is bounded by a line $\ell$, which lies in the horizontal plane $z=10$ and thus does not contain the origin.
To construct the curtains,  let $A$ be the plane through $\ell$ and the origin, and let $\vec d$ be the normal vector of $A$ pointing in the direction away from the interior of $H$. Then $({\bf 0}, \vec d)$ is the curtain for $H$.\footnote {Note that all curtains are constructed passing through the origin; later they will be translated.}
See Fig.~\ref {fig:curtains-0}.

\begin {observation}
A curtain has a single vertex.
This vertex is
\new{a local maximum} in all directions in the corresponding half-plane $H$,
and is not a local maximum in all directions outside $H$.
\end {observation}

%\maria{I guess we need some kind of if and only if property here, something like: and it is not maximal in any of the directions corresponding to points in $Q$ outside the half-plane?}
%\maarten{Agreed. I made it into an observation above.}
%\frank{maybe show a direction $d'$ in the half-plane, in Fig 13 a,b,c?}

Next, we color the curtains with a color out of the set $\{TL, TR, BL, BR\}$, \new{standing for \textit{top-left}, \textit{top-right}, \textit{bottom-left}, and \textit{bottom-right}, respectively, depending on the orientation of the associated half-plane}, as follows.
First, we discard all half-planes that do not not intersect $Q$ at all (since they will never contribute to a solution anyway).
For the remaining half-planes, those bounded by a line with a positive slope get assigned to $TL$ or $BR$, and those bounded by a line with a negative slope get assigned to $BL$ or $TR$, such that all half-planes with color $TL$ contain the top left corner of $Q$, all half-planes with color $TR$ contain the top right corner of $Q$, all half-planes with color $BL$ contain the bottom left corner of $Q$, and all half-planes with color $BR$ contain the bottom right corner of $Q$.
See Fig.~\ref {fig:half-planes}c for an example.
The curtains are colored correspondingly.

By assigning these colors to the curtains, we have grouped them into four groups whose rays point in similar directions. Specifically, we have the following property.

\begin {lemma} \label {lem:hitplanes}
  All curtains with color $TR$ that start above the plane $S: x + y + z = 1$ have a direction ray that intersects this plane.
  More specifically, a curtain $((\frac12, \frac12, \frac14), \vec d)$ with color $TR$ has a ray that hits the plane somewhere inside the trapezium $\nabla$ bounded by the points 
 $(\frac12, \frac14 - \frac{1 + 5\sqrt2}{196}, \frac14 + \frac{1 + 5\sqrt2}{196}),
  (\frac12, \frac14 + \frac{-1 + 5\sqrt2}{196}, \frac14 - \frac{-1 + 5\sqrt2}{196}),  
  (\frac14 + \frac{-1 + 5\sqrt2}{196}, \frac12, \frac14 - \frac{-1 + 5\sqrt2}{196}), \mathrm {and}\ 
  (\frac14 - \frac{1 + 5\sqrt2}{196}, \frac12, \frac14 + \frac{1 + 5\sqrt2}{196})
 $. 
\end {lemma}

\begin {proof}
  Consider a curtain $(p,\vec d)$ of color TR, that starts above the plane $S: x + y + z = 1$. Refer to Fig.~\ref {fig:curtains}. Since the curtain has color $TR$, the associated half-plane is bounded by a line with a negative slope, so it has equation $x+By+Cz=0$, with $B>0$. Since the half-plane contains the top left corner of $Q$, the normal vector $\vec d$ of the associated plane through the origin has equation $(-1,-B,-C)$. Hence, the ray bounding $(p,\vec d)$ intersects the plane $S$.
  
  Consider the infinite bundle $\beta$ of all possible curtains of color $TR$ that start in the point $p= (\frac12, \frac12, \frac14)$. We now argue that any curtain in this bundle intersects $S$ inside the trapezium $\nabla$.  
  Clearly, if the rays on the boundary of the bundle $\beta$ have this property, the interior rays will have it too.
  By the above argument, the vertical projection of the ray from $p$ in direction $\vec d$ makes an angle with the ray in direction $(-1, -1)$ of at most $45^\circ$.  
  %Since the curtain has color $TR$, the vertical projection of the ray from $p$ in direction $\vec d$ makes an angle with the ray in direction $(-1, -1)$ of at most $45^\circ$.
  %\rodrigo{Why does TR imply this? Not obvious to me}\maarten{This is exactly by definition. In fact, we might define the colours explicitly based on the ray directions, and then just state the containment of corners of $Q$ as a consequence, given the somewhat messy statement we get otherwise (see my comment in Figure 10).}\maria{Can we instead add more details? For example, what about: 
  %}\rodrigo{I do appreciate having some more details!}. 
  This means that, the extremal rays corresponding to the curtains in $\beta$ are parallel to the $x$-axis or the $y$-axis.
  %\maria{So it is enough to look at the extreme cases? Can we say it somehow?} \maarten {What do you mean by "it is enough"? 
  %We have a bundle of rays and want to argue that they all intersect $\nabla$.
  %a certain plane (or a certain convex region in that plane). 

  %} \maria{OK. Can you please state this in a more formal way and add it? I guess here it is important that there is some restriction on the $C$ of the equation above?}

  Let us consider the case where it is parallel to the $x$-axis; the other case is symmetric.
  If the projected ray is parallel to the $x$-axis, $\vec d$ is parallel to the $xz$-plane, and thus stays in the plane $y=\frac12$. In this plane, the plane $S$ corresponds to the line $x + z = \frac12$. 

  Now, since we lifted $Q$ to the height $z=10$, in the worst case, the vertical angle of the ray in direction $\vec d$ has a slope between $\frac{\sqrt2}{10}$ and $-\frac{\sqrt2}{10}$.\footnote{Note that, if the ray is parallel to the $x$ or $y$ axis in the vertical projection, the slope is actually restricted to a smaller interval between $\frac{1}{10}$ and $-\frac{1}{10}$, so the true region where it may intersect $S$ is even smaller than the trapezium we calculate here (but the true region is not a polygon, so we use the larger estimate for simplicity).}
  %\maria{How did you obtain these numbers?}\maarten{Exactly my point. You don't want to worry about that as a reader. I don't remember the details myself after a few months, but it should be what you get when you calculate the intersection coordinates in the side view (i.e. Figure 12(c)).} \maria{as I wrote in the email, I would add the details if you are able to reconstruct them}
  %\rodrigo{Maybe a compromise solution would be to explain that these numbers come from the intersection of X and Y (as Maarten was explaining), and leave it there. In that way, it is easier for the reader to reconstruct the numbers, but we don't spend much space spelling out all details}
A ray from the point $(\frac12, \frac14)$ with slope $\frac{\sqrt2}{10}$ intersects the line $x+z=\frac12$ in the point $q = (\frac14-\delta,\frac14+\delta \frac{\sqrt2}{10})$ which solves to $\delta = \frac{1 + 5\sqrt2}{196} \approx 0.041$ so $q \approx (0.209, 0.281)$.
A ray from the point $(\frac12, \frac14)$ with slope $-\frac{\sqrt2}{10}$ intersects the line $x+z=\frac12$ in the point $q' = (\frac14+\delta',\frac14-\delta \frac{\sqrt2}{10})$ which solves to $\delta = \frac{-1 + 5\sqrt2}{196} \approx 0.031$ so $q \approx (0.281, 0.219)$.\qed
%\maarten {Not sure how much detail we want to keep here.}
\end {proof}

Symmetric to Lemma~\ref{lem:hitplanes} the rays bounding the curtains of color $TL$ intersect a small trapezium on the plane $-x+y+z = 1$, those of color $BL$ intersect a trapezium on the plane $-x -y + z = 1$,  and the rays bounding curtains of color $BR$ intersect a trapezium on the plane $x - y + z = 1$.

\begin{figure}[tb]
\begin{center}
    \includegraphics[page=5]{3sum_hard_new_proof.pdf}
  \end{center}
  \caption
  { (a) A curtain $(p, \vec d)$ in $TR$ will intersect the plane $x + y + z = 1$, and more specifically, when $p=(\frac12, \frac12, \frac14)$ the intersection point lies in the trapezium bounded by the four points
 $(\frac12, \frac14 - \frac{1 + 5\sqrt2}{196}, \frac14 + \frac{1 + 5\sqrt2}{196}),
  (\frac12, \frac14 + \frac{-1 + 5\sqrt2}{196}, \frac14 - \frac{-1 + 5\sqrt2}{196}),  
  (\frac14 + \frac{-1 + 5\sqrt2}{196}, \frac12, \frac14 - \frac{-1 + 5\sqrt2}{196}), \mathrm {and}
  (\frac14 - \frac{1 + 5\sqrt2}{196}, \frac12, \frac14 + \frac{1 + 5\sqrt2}{196})
 $
    (b) Top-down view.
    (c) Side view. \new{The axis denoted by $x+y$ is the axis in the direction (1,1,0), i.e., a side view with the camera position rotated by $45^\circ$ so that the green triangle becomes a line segment.}
    %\rodrigo{Not sure what $x+y$ means in the axis}.\maarten {It's the axis in the direction (1,1,0); i.e., a side view with the camera position rotated by 45 degrees so that the green triangle becomes a line segment.}
  }
  \label{fig:curtains}
\end{figure}

Note that the shape of the trapezium from Lemma~\ref {lem:hitplanes} changes linearly with the location of $p$; i.e., if we move $p$ closer to the side of the plane, the trapezium will shrink, and if we move it parallel to the plane, the trapezium will translate.

%\frank{plan: 1) put the plane at $z=10$; so that the direction vectors of
%the curtains are more horizontal. 2) argue here that they intersect
%the side o four pyramid in some little rectangular region. 3) add four
%additional halfspaces that guarantee that we consider only directions
%corresponding to points in $Q$.
%}

\subsubsection {Building a terrain.}

In this section, we will build a degenerate terrain, by placing all of our curtains on top of a mountain.
The terrain is degenerate because it has vertical faces (the sides of the curtains) but will otherwise have exactly the properties that we require.

\begin{figure}[tb]
\begin{center}
    \includegraphics[page=4]{3sum_hard_new_proof.pdf}
  \end{center}
  \caption
  { (a) The base of the terrain is a pyramid with four sides in the colors $\{TL, TR, BL, BR\}$.
    (b) We place the curtains on the sides of the pyramid with the same color.
    (c) Top-down view of the terrain.
  }
  \label{fig:terrainconstruction}
\end{figure}

The base of our terrain will be a pyramid whose sides lie on the four planes $x+y+z=1$, $x-y+z=1$, $-x+y+z=1$ and $-x-y+z=1$ (i.e.~the plane from Lemma~\ref{lem:hitplanes} and the corresponding symmetric planes). See Fig.~\ref {fig:terrainconstruction}a.
%.\maria{Technically, Lemma~\ref {lem:hitplanes} speaks only about one plane; what about: 

%First, we argue that the top of the pyramid is a local maximum exactly in any direction inside $Q$, and not for the directions outside $Q$. This means that the prickliness of our terrain can be equal to $n+1$ only for directions in $Q$, as desired.

%\begin {lemma}
%  The top of the pyramid is a local maximum if and only if we look in a direction in $Q$.
%\end {lemma}
%\frank{observation?}

\begin {observation} \label {obs:sharp-top}
  The top of the pyramid is a local maximum for every direction in $Q$.
\end {observation}

%First, we observe that the top of the pyramid is always a local maximum for any direction inside $Q$.
%\rodrigo{What are directions in $Q$? Isn't $Q$ a square?}\maarten{Yes, it is a square that contains points, and the directions are the directions from the origin to those points. Admittedly this is a slight abuse of terminology, but we use it a lot throughout this section (we talk about "directions in H" and "directions outside H" where H is a halfplane all the time). Do you think we need to introduce explicit notation / terminology instead? For example, we could introduce an explicit function $\zeta$ that turns points into directions, and replace $Q$ by $\zeta(Q)$ in the sentence above - would it be more clear or more distracting?}\maria{I would do nothing here, as Maarten mentions we use this abuse a lot in the section}
%\rodrigo{OK, let's leave it, but perhaps we should add some phrase at the beginning warning the reader that we will abuse notation a little bit?}.  
This observation implies that the number of local maxima of our terrain in a fixed direction $\vec d$ will be $1$ higher than the number of half-planes that $\vec d$ is contained in for all directions $\vec d$ in $Q$. 
(Recall that for directions outside $Q$, the number of local maxima can be at most $n$, so they are irrelevant for our reduction.

%\rodrigo{Didn't we say that we discard these? Or maybe I am confused with something else outside $Q$ that we discard?}\maarten {Yes, we discarded half-planes that lie outside $Q$, but this argument is about any potential direction that might define the prickliness})

%\rodrigo{The previous observation looks strange after we observed the same thing in the paragraph above, does it add much?}

Next, we attach the curtains to the sides of the pyramid.
For this, we place the vertices of the curtains close to and above the
sides of the pyramid and sufficiently far from each other, so they do
not interfere with each other, i.e., they do not intersect each other before they intersect the pyramid. By Lemma~\ref {lem:hitplanes}, \new{if we place the vertices of the curtains close enough to the
sides of the pyramid,} the rays of the curtains
all intersect the side of the pyramid, and clearly the vertical ray
down will also hit the side of the pyramid. %\frank{well, they  intersect a translated copy of the *supporting plane*. That they intersect the side must be due to how close we place the vertices of the curtains.}
\new{Finally, we clip the curtains so that the ray
does not stick out through the other side.}
%\frank{I guess we clip the curtains so that the ray doesn't come out the other side.}
See Fig.~\ref {fig:terrainconstruction}b and Fig.~\ref {fig:terrainconstruction}c.

\begin {lemma}
  The degenerate terrain $T$ we have constructed has prickliness $n+5$ if and only if there is a point in $Q$ that is {\em not} contained in any of the $n$ strips.
\end {lemma}

\begin {proof}
  By construction, our terrain $T$ \new{consists of the base pyramid, with 5 vertices, and $2n+4$ curtains}. Thus, it has (at most) $3(2n+4) + 5$ vertices\footnote {Or less, if some half-planes were deleted because they did not intersect $Q$.}. Of these, $4n+8$ vertices are concave (where the rays or vertical rays of the curtains intersect the pyramid), so these will never be local maxima. The four external vertices are not considered in our definition of prickliness, and the fifth is the top of the pyramid, which is always a local maximum by Observation~\ref{obs:sharp-top}. 
  So that leaves $2n + 5$ vertices \new{that can be local maxima}.

  \new{Out of the $2n+4$ curtain vertices, $2n$ come in pairs, corresponding} to the two half-planes bounding the same strip. The set of directions in those two half-planes is disjoint, so only one \new{from each} pair can be a local maximum for a given direction. 
  That means that the prickliness can never be more than $n+5$.

  Now, if the number of local maxima is equal to $n+5$ for some direction $\vec d$, that means the top of the pyramid is a local maximum and that all four curtains corresponding to half-planes bounding $Q$ are also a local maximum (so $\vec d$ is in $Q$), and one of each pair of curtains is also a local maximum ($\vec d$ is in one of the half-planes bounding a strip, and thus, not contained in the strip). So $\vec d$ is a witness to the \textsc {strips-cover-box} being a NO-instance. The argument works both ways: if we have a NO-instance, then the prickliness will be exactly $n+5$.
\end {proof}

\subsubsection {Coordinates, precision, and degeneracy.}

We need to specify the coordinates of the curtains to ensure they do not interfere. It is essential that we do so without needing to explicitly test their interactions 
%\rodrigo{Interference, interaction, it's all a bit mysterious... aren't we talking about intersections?}\maarten {Because the curtains do intersect each other (well some of them do) underneath the pyramid, or beyond it, but the point is that we don't want to test for those intersections, because we don't care about them. But I agree it is somewhat vaguely phrased, so open to suggestions for improvement. :)}\maria{As suggested by Frank, I added above: Finally, we clip the curtains so that the ray
%does not come out the other side. With this addition, the curtains do not intersect anymore underneath the pyramid, or beyond it, so we can replace interactions by intersections?} 
(as we cannot afford to spend time on checking all pairs). So, we place the curtains so that the trapeziums as in Lemma~\ref {lem:hitplanes} are sufficiently small. 

Let $H_1, H_2, \ldots H_m$ be our set of half-planes of color $TR$. We place the apices of the curtains at coordinates $(\frac {i}{m+1}, \frac {m+1-i}{m+1}, \frac {1}{4m})$. Then the resulting trapeziums from Lemma~\ref {lem:hitplanes} are scaled down by a factor $\frac1m$ and translated to be disjoint on the surface of the pyramid. We handle the other three colors symmetrically. 

Note that the precision of coordinates required in the reduction is not significantly greater than the precision of the input numbers, which implies that our reduction also holds in other models of computation than the Real RAM.

Finally, the above construction results in a degenerate terrain. In order to lift the degeneracy, we replace each curtain by a slightly expanded curtain as follows. For a curtain $(p, \vec d)$, we still create a vertex at the point where the ray in direction $\vec d$ intersects the pyramid, but we no longer create a vertex directly below $p$. Instead, we shoot two rays from $p$ almost straight down, resulting in {\em two} vertices, one on each side of the curtain. For curtains with color $TR$, the two downward rays are in directions $(0, 1, -100)$ and $(1, 0, -100)$; the other colors are handled symmetrically. Note that this does influence the set of directions in which $p$ is a local maximum: it is no longer a half-plane, but the intersection between three half-planes (i.e., a triangle), but since the other two sides of the triangle do not intersect $Q$, the region that is inside such a triangle {\em and} within $Q$ remains the same as it was for the half-plane.

\begin{theorem}
  The problem of computing the prickliness of a 2.5D terrain is 3SUM-hard. Specifically:
  \[
    \textsc{3sum} \lll_{n \log n} \textsc{prickliness}.
  \]
\end{theorem}

\begin {proof}
  We have just shown that $\textsc{strips-cover-box} \lll_{n} \textsc{prickliness}$.
  The result now follows from Theorem~\ref {thm:strips-cover-box}, which state that $\textsc{3sum} \lll_{n \log n} \textsc{strips-cover-box}$.
\end {proof}

\part{Experiments}

\section{Existing topographic attributes}
\label{sec:Topographic_attributes}
%\maarten {This seems very late into the paper to introduce these measures. Shouldn't we discuss these in the introduction, or maybe preliminaries, before any of our actual results?}
%\rodrigo{That would be fine too, and has the advantage that then it will be easier for the reader to understand that these are not affine-invariant}
%\maarten {Moving the section for now, to see how it looks.}

% \rodrigo{We still should say a word about why we picked these (and not
%   other) attributes}\frank{I think those were all that seemed remotely
% related to viewshed complexity.}\maria{What about: In addition to prickliness, in our experiments we consider other existing topographic
% attributes that aim to describe the shape of a terrain and thus might be related to viewshed complexity.}

Since one of the goals of our experiments is to verify our hypothesis that existing topographic attributes do
  not provide a good indicator of the viewshed complexity, in our experiments we consider the following topographic
attributes (in addition to prickliness).

%In addition to prickliness, we consider the following topographic
%attributes, and their relation to viewshed complexity.

\newcommand{\pTRI}{\ensuremath{\mathit{pTRI}}\xspace}
\newcommand{\TRI}{\ensuremath{\mathit{TRI}}\xspace}
\newcommand{\pTSI}{\ensuremath{\mathit{pTSI}}\xspace}
\newcommand{\TSI}{\ensuremath{\mathit{TSI}}\xspace}
\newcommand{\FD}{\ensuremath{\mathit{FD}}\xspace}
\newcommand{\pFD}{\ensuremath{\mathit{pFD}}\xspace}
\newcommand{\pSKI}{\ensuremath{\mathit{pSKI}}\xspace}
\newcommand{\SKI}{\ensuremath{\mathit{SKI}}\xspace}

\paragraph{Terrain Ruggedness Index (TRI)} The Terrain Ruggedness
Index measures the variability in the elevation of the
terrain~\cite{Riley1999TRI}. Riley \etal originally defined TRI
specifically for DEM terrains as follows. Let $c$ be a cell of the terrain, and let $\mathcal{N}(c)$ denote the set of (at most) eight neighboring cells of $c$.
The TRI of $c$ is then defined as
 \[  \pTRI(c) = \sqrt{\frac{1}{\lvert \mathcal{N}(c)\rvert}\sum_{q \in \mathcal{N}(c)} \left(c_z - q_z\right)^2},\]
where $ \lvert \mathcal{N}(c)\rvert$
denotes the cardinality of $\mathcal{N}(c)$. Hence, $\pTRI(c)$ essentially
measures the standard deviation of the difference in height between
$c$ and the points in $\mathcal{N}(c)$. The Terrain
Roughness Index $\TRI(T)$ of $T$ is the average
$\pTRI(T,c)$ value over all cells $c$ in $T$. 

For TINs, we have adapted the definition as follows: $\mathcal{N}(c)$ is defined as the set of vertices which are adjacent to a given vertex $c$ in $T$. The Terrain Roughness Index $\TRI(T)$ is then obtained as the average
$\pTRI(T,c)$ value over all vertices $c$ of $T$.

\paragraph{Terrain Shape Index (TSI)} The Terrain Shape Index also
measures the ``shape'' of the terrain~\cite{McNab1989TSI}. 
Let
$\mathcal{C}(c,r,T)$ denote the intersection of $T$ with a vertical
cylinder of radius $r$ centered at $c$ (so after projecting all points
to the plane, the points in $\mathcal{C}(c,r,T)$ lie on a circle
of radius $r$ centered at $c$). For ease of computation, we discretize
$\mathcal{C}(c,r,T)$: for DEMs we define $C(c,r,T)$ to be the
grid cells intersected by $\mathcal{C}(c,r,T)$, and for TINs
we define $C(c,r,T)$ as a set of 360 equally spaced points on
$\mathcal{C}(c,r,T)$. The TSI of a point $c$ is then defined as
\[
  \pTSI(c,r,T) = \frac{1}{r	\lvert C(c,r,T)	\rvert } \sum_{q \in C(c,r,T)} c_z -
  q_z\ ,
\]
and essentially measures the average difference in height between
``center'' point $c$ and the points at (planar) distance $r$ to $c$,
normalized by $r$. The Terrain Shape Index $\TSI(T,r)$ of the entire
terrain $T$ is the average $\pTSI(c,r,T)$ over all cells (in case of a
DEM) or vertices (in case of a TIN) of $T$. We choose $r=1000$m (which
is roughly eight percent of the width of our terrains) in our
experiments.

\paragraph{Fractal dimension (FD)} The (local) fractal dimension
measures the roughness around a point $c$ on the terrain over various
scales~\cite{Mandelbrot1982,Taud2005Fractal}. We use the definition of
Taud and Parrot~\cite{Taud2005Fractal} that uses a box-counting
method, and is defined as follows. Let $w$ be the width of a cell in
the DEM, and let $s \in \mathbb{N}$ be a size parameter. For
$q \in 1..s/2$, consider subdividing the cube with side length $sw$
centered at $c$ into $(s/q)^3$ cubes of side length $qw$. Let
$\mathcal{C}_s(c,q)$ denote the resulting set of cubes, and define
$N_s(c,q,T)$ as the number of cubes from $\mathcal{C}_s(c,q)$ that
contain a ``unit'' cube from $\mathcal{C}_s(c,1)$ lying fully below
the terrain $T$. Let $\ell_s(c,T)$ be the linear function that best
fits (i.e. minimizes the sum of squared errors) the set of points
$\{(\ln(q),\ln(N_s(c,q,T))) \mid q \in 1..s/2\}$ resulting from those
measurements. The fractal dimension $\pFD_s(c,T)$ at $c$ is then
defined as the inverse of the slope of $\ell_s(c,T)$. The fractal
dimension $\FD(T,s)$ of the DEM terrain itself is again the average
over all DEM cells. Following Taud and Parrot we use $s=24$ in our
experiments. For our TIN terrains, we keep $w$ the same as in their
original DEM representations, and average over all vertices.

\section{Experiments}
\label{sec:experiments}

In this section we present our experimental setup. Our goals are to
\begin{itemize}
\item verify our hypothesis that existing topographic attributes do
  not provide a good indicator of the viewshed complexity,
\item evaluate whether prickliness \emph{does} provide a good
  indicator of viewshed complexity in practice, and
\item evaluate whether these results are consistent for DEMs and TINs.
\end{itemize}
Furthermore, since in many applications we care about the visibility
of multiple viewpoints (e.g. placing guards or watchtowers), we also
investigate these questions with respect to the complexity of the
common viewshed of a set of viewpoints. In this setting a point is
part of the (common) viewshed if and only if it can be seen by at
least one viewpoint. Note that since Theorem~\ref{thm:2.5comp}
\emph{proves} that the complexity of a viewshed is proportional to the
prickliness, our second goal is mainly to evaluate the practicality of prickliness. That is, to establish if this relation
is also observable in practice or that the hidden constants in the
big-O notation are sufficiently large that the relation is visible
only for very large terrains.

Next, we briefly describe our implementations of the topographic
attributes. We then outline some basic information about the terrain
data that we use as input, and we describe how we select the
viewpoints for which we compute the viewsheds.

% \subsection{Implementations}
% \label{sub:Implementations}
\paragraph{Implementations} We consider prickliness\footnote{For technical reasons, our implementation considers a vertex to be a local maximum if all its neighboring vertices are lower than the vertex (rather than at a height that is lower than or equal to that of the vertex), but that does not make a difference for the terrains considered.} and the
topographic attributes from
Section~\ref{sec:Topographic_attributes}. To compute the prickliness
we implemented the algorithm from Theorem~\ref{thm:2D} in C++
using CGAL~5.0.2~\cite{CGAL:CGAL} and its \emph{2D
arrangements}~\cite{CGAL:2dArrangements} library.\footnote
{ Source code available from \url {https://github.com/GTMeijer/Prickliness}. }
We also implemented
the algorithms for TRI, TSI, and FD on TINs. These are mostly
straightforward. To compute the viewsheds on TINs we implemented the
hidden-surface elimination algorithm of~\cite{Goodrich1992Viewshed} using CGAL.\footnote
{ Source code available from \url {https://github.com/GTMeijer/TIN_Viewsheds}.}
We remark that our implementation computes the exact TIN viewsheds, as opposed to previous studies that only considered fully visible triangles (e.g.,~\cite{RD-CausesError-07}).
% This algorithm can
% then also be used to determine the part $\mathcal{V}(c,T)$ of
% $\mathcal{B}(c)$ visible from a point $c$, and thus implement the sky
% visibility index.

To compute the prickliness on DEMs we used the algorithm from Section~\ref{sub:Algorithm_for_DEMs}. For TRI, TSI,
and FD on DEMs we used the implementations available in
ArcGIS~Pro~2.5.1~\cite{ArcGISPro}. % The sky visibility index can be computed
% using the skyline tool.
To compute viewsheds on DEMs we used the builtin tool ``Viewshed 2'' in ArcGIS with a vertical offset of 1 meter, which produces a raster with boolean values that indicate if a cell is visible or not. To get a measure of
complexity similar to that of the TINs we use the ``Raster to
Polygon'' functionality of ArcGIS (with its default settings) to convert the set of \textsc{True}
cells into a set of planar polygons (possibly with holes). We use the
total number of vertices of these polygons as the complexity of the
viewshed on a~DEM.
% \frank{Maybe some small figure of what this looks like would be nice?} \maria{I am afraid it's a little late for this :)}

% \subsection{Terrains}
% \label{sub:Terrains}
\paragraph{Terrains} We considered a collection of 43 real-world
terrains around the world.  These terrains were handpicked in order to
cover a large variety of landscapes, with varying ruggedness,
including mountainous regions (Rocky mountains, Himalaya), flat areas
(farmlands in the Netherlands), and rolling hills (Sahara), and different
complexity. We obtained the terrains from the United States Geological Survey (USGS),
and converted them through the \emph{Terrain} world
elevation layer \cite{ESRITerrain} in  ArcGIS~\cite{ArcGISPro};
all terrains use the WGS 1984 Web Mercator (auxiliary sphere) map projection. 
Each terrain is represented as a 10-meter resolution DEM of size  1400$\times$1200. According to past studies the chosen resolution
of 10 meters provides the best compromise between high resolution and
processing time of measurements~\cite{Maynard2014DEMRes,Zhang1994DemRes}. The complete list of terrains with their extents
can be found in~\cite{thesis/gert}.

We generated a TIN terrain for each DEM using the \emph{``Raster to
  TIN''} function in ArcGIS~\cite{ArcGISPro}. This function
generates a Delaunay triangulation to avoid long, thin triangles as
much as possible.  With the \emph{$z$-tolerance} setting, the
triangulation complexity can be controlled by determining an allowed
deviation from the DEM elevation values. We considered TIN terrains
generated using a $z$-tolerance of 50 meters. This resulted in TINs
where the number of vertices varied between 30 and 3304 (with an
average of 1125 vertices). Their distribution can be seen in Fig.~\ref{fig:terrains_prickliness}.

\begin {figure}[tb]
  \centering
  \includegraphics[clip,trim=0 0 0 0.7cm]{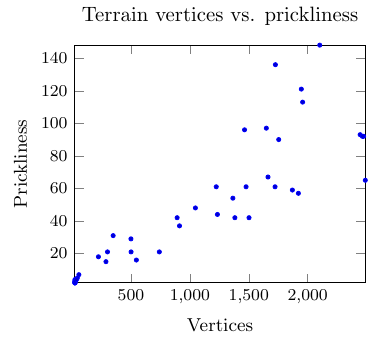}
  \caption {The prickliness values for the terrains we considered.}
  \label{fig:terrains_prickliness}
\end {figure}

% Initially, a $z$-tolerance of 100 meters was used.  This
% $z$-tolerance generated TIN terrains that ensured reasonable
% processing times.  To further explore certain measures' behavior and
% come closer to the detail level of the DEM terrains, a TIN terrain set
% with a $z$-tolerance of 50 meters was also generated.

% \subsection{Viewpoints}
% \label{sub:Viewpoints}

\paragraph{Viewpoints} Kim \etal~\cite{Kim2004ViewpointLocations}
found that placing the viewpoints at peaks typically produces
viewsheds that cover hilltops, but not many valleys, whereas placing
viewpoints in pits typically covers valleys but not hilltops. 
This leads us to consider three different strategies to pick the locations
of the viewpoints: picking ``high'' points (to cover peaks), picking
``low'' points (to cover valleys), and picking viewpoints uniformly at
random. To avoid clusters of high or low viewpoints we overlay an
evenly spaced grid on the terrain, and pick one viewpoint from every
grid cell (either the highest, lowest, or a random one). We pick
these points based on the DEM representation of the terrains, and
place the actual viewpoints one meter above the terrain to avoid
degeneracies. We use the same locations in the TINs in order to
compare the results between TINs and DEMs (we do recompute the
$z$-coordinates of these points so that they remain 1m above the surface of the TIN). 
The resulting viewsheds follow the expected pattern; refer to Fig.~\ref {fig:coverage}.
In our experiments, we consider both
the complexity of a viewshed of a single viewpoint as well as the
combined complexity of a viewshed of nine viewpoints (picked from a
$3\times3$ overlay grid). Results of Kammer \etal~\cite{klms-papgpt-14}
suggest that for the size of terrains considered these viewpoints
already cover a significant portion of the terrain, and hence picking
even more viewpoints is not likely to be informative.

\begin{figure}[h]%
	\centering
	\includegraphics [angle=90,scale=0.4] {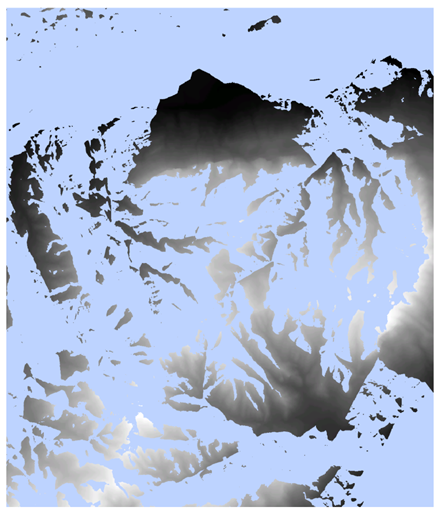}
	\includegraphics [angle=90,scale=0.4]{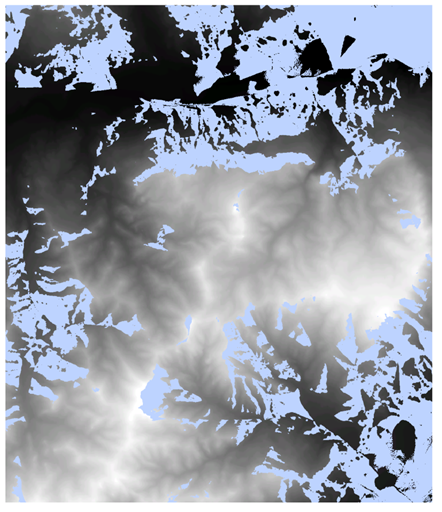}
	%\hspace{8pt}%
	\caption{
		(left) A joint viewshed (blue) created from viewpoints placed on the highest points. (right) A joint viewshed (blue) created from viewpoints placed on the lowest points.
	}%
	\label{fig:coverage}%
\end{figure}

% \subsection{Analysis}
% \label{sub:Analysis}
\paragraph{Analysis} For each  topographic attribute we
consider its value in relation to % to the complexity of the terrains
% considered, and to 
the complexity of the viewsheds.
In addition, we
test if there is a correlation between the viewshed complexity and the
attribute in question. We compute their sample correlation coefficient
(Pearson correlation coefficient) $R$ to measure their (linear)
correlation. The resulting value is in the interval $[-1,1]$, where a
value of $1$ implies that a linear increase in the attribute value
corresponds to a linear increase in the viewshed complexity. A value
of $-1$ would indicate that a linear increase in the attribute leads
to a linear decrease in viewshed complexity, and values close to zero
indicate that there is no linear correlation. % We also report the
% squared coefficients (indicated by $r^2$) to express .....
%\frank{say s.t. about the $R^2$ scores; not entirely sure what to say
%  about them though, other than that being close to one is good. I
%  don't think they add much w.r.t the $r$ values we report.}

\section{Results}
\label{sec:Results}

We start by investigating the prickliness values compared to the
complexity of the terrains considered. These results are shown in
Fig.~\ref{fig:terrains_prickliness}. We can see that the prickliness
is generally much smaller than the number of vertices in the (TIN
representation of the) terrain. In Fig.~\ref{fig:real} we also see
the $\pi_{\vec{v}}$ values for orientation vectors near $(0,0,1)$
(recall that the maximum over all orientations defines the
prickliness).

\begin {figure}[tb]
  \centering
  \includegraphics [scale=.72] {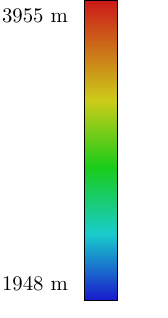}
  \includegraphics [scale=.14] {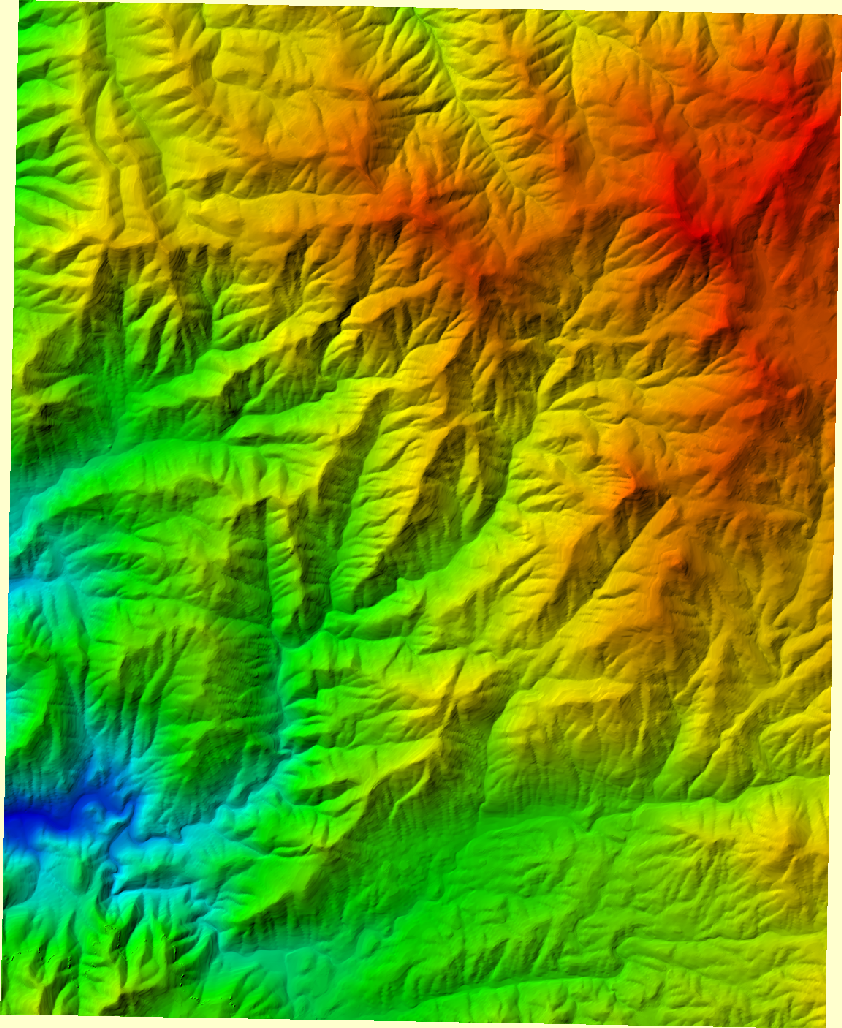}
  \qquad
  \includegraphics [scale=.72] {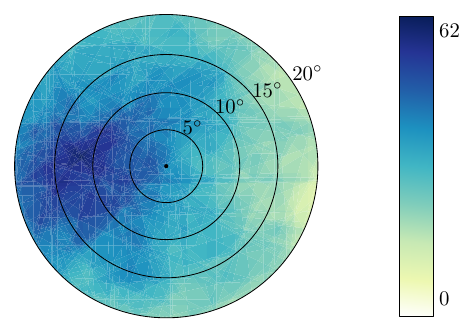}
  \caption {(left) A real-world terrain with $583$ vertices from the
    neighborhood of California Hot Springs whose prickliness is only
    $62$. (right) The value $\pi_{\vec{v}}$ for vectors near
    $(0,0,1)$.}
  \label{fig:real}
\end {figure}

%\subsection{Topographic attributes and viewshed complexity}
%\label{sub:Topographic_attributes_and_viewshed_complexity}

Next, we analyze the relation between topographic attributes and viewshed complexity. For each of the terrains we compute the viewshed of one or nine
viewpoints, for three viewpoint placement strategies, and analyze the
complexity of the viewshed as a function of the topographic attributes
both for TINs and DEMs. Table~\ref{table:combined} summarizes the 
correlation between the attributes and the viewshed complexity for
each case.

 \setlength{\tabcolsep}{4.2pt}
	\begin{table}[tb]
		\begin{tabular*}{\linewidth}{l *{12}{r} }
          \toprule
            & \multicolumn{6}{c}{TIN} & \multicolumn{6}{c}{DEM}\\
			& \multicolumn{2}{c}{highest} 
			& \multicolumn{2}{c}{lowest} 
			& \multicolumn{2}{c}{random} 
			& \multicolumn{2}{c}{highest} 
			& \multicolumn{2}{c}{lowest} 
			& \multicolumn{2}{c}{random} \\
			\cmidrule(lr){2-3}\cmidrule(lr){4-5}\cmidrule(lr){6-7}
			\cmidrule(lr){8-9}\cmidrule(lr){10-11}\cmidrule(lr){12-13}
			& \centercell{\single {single}}
			& \centercell{\multi  {multi}} 
			& \centercell{\single {single}}
			& \centercell{\multi  {multi}} 
			& \centercell{\single {single}}
			& \centercell{\multi  {multi}} 
			& \centercell{\single {single}}
			& \centercell{\multi  {multi}} 
			& \centercell{\single {single}}
			& \centercell{\multi  {multi}} 
			& \centercell{\single {single}}
			& \centercell{\multi  {multi}} \\
			\cmidrule(lr){2-7}\cmidrule(lr){8-13}
			Prick & \single{ 0.75} & \multi{ 0.97} & \single{ 0.41} & \multi{ 0.83} & \single{ 0.64} & \multi{ 0.93} 
		  		  & \single{ 0.62} & \multi{ 0.90} & \single{ 0.09} & \multi{ 0.18} & \single{ 0.16} & \multi{ 0.63} \\
			TRI   & \single{ 0.45} & \multi{ 0.58} & \single{ 0.69} & \multi{ 0.72} & \single{ 0.58} & \multi{ 0.66} 
      			  & \single{-0.53} & \multi{-0.37} & \single{-0.30} & \multi{-0.35} & \single{-0.37} & \multi{-0.42} \\
			TSI	  & \single{ 0.47} & \multi{ 0.66} & \single{ 0.75} & \multi{ 0.79} & \single{ 0.58} & \multi{ 0.74} 
		  		  & \single{-0.52} & \multi{-0.39} & \single{-0.26} & \multi{-0.28} & \single{-0.33} & \multi{-0.42} \\
			FD    & \single{-0.56} & \multi{-0.71} & \single{-0.68} & \multi{-0.80} & \single{-0.61} & \multi{-0.77} 
		  		  & \single{ 0.02} & \multi{-0.48} & \single{ 0.27} & \multi{ 0.23} & \single{ 0.24} & \multi{-0.18} \\
			\bottomrule
		\end{tabular*}
		\caption{The correlation coefficients ($R$ values) between the attributes and viewshed complexity.}
		\label{table:combined}
	\end{table}

\paragraph{TIN results}

In this case we distinguish between single and multiple viewpoints.

For single viewpoints, the first row in
Fig.~\ref{fig:TIN_results} shows the full results for a randomly
placed viewpoint on the TIN. Somewhat surprisingly, we see that
terrains with high fractal dimension have a low viewshed
complexity. For the other measures, higher values tend to
correspond to higher viewshed complexities. However, the scatter plots
for TRI and TSI show a large variation.
The scatter plots for the other placement
strategies (highest and lowest) look somewhat similar (see Fig.~\ref{fig:TIN_results_highest} and~\ref{fig:TIN_results_lowest} in Appendix~\ref{sec:app}), hence the strategy
with which we select the viewpoints does not seem to have much
influence in this case. None of the four attributes shows a strong
correlation in this case (see also Table~\ref{table:combined}).
Prickliness shows weak-medium correlation in three out of six cases,  strong correlation for one case---viewpoints at highest points---and no correlation for two cases with viewpoints at lowest points.
The other attributes show an even weaker correlation in general.

\begin{figure}[tb]
  % \centering
  % \includegraphics[width=0.3\textwidth]{plots/PricklyPlots/Highest/Single_50_Highest_Prickly_TIN}
  % \includegraphics[width=0.3\textwidth]{plots/PricklyPlots/Lowest/Single_50_Lowest_Prickly_TIN}
  %
  % \includegraphics[width=0.3\textwidth]{MeasureUpdates/Plots/TIN_1VP/1VP_Highest_TRI_50_TIN}
  % \includegraphics[width=0.3\textwidth]{MeasureUpdates/Plots/TIN_1VP/1VP_Lowest_TRI_50_TIN}
  %
  % \includegraphics[width=0.3\textwidth]{MeasureUpdates/Plots/TIN_1VP/1VP_Highest_TSI_50_TIN}
  % \includegraphics[width=0.3\textwidth]{MeasureUpdates/Plots/TIN_1VP/1VP_Lowest_TSI_50_TIN}
  \hspace{-1em}
  \begin{tabularx}{\linewidth}{cccc}
    % Viewshed Complexity
    ~~~~~Prickliness & ~~~FD & ~~TRI & ~~TSI \\
      \includegraphics[height=2.2cm]{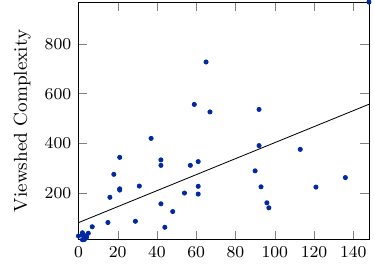}
	& \includegraphics[height=2.2cm]{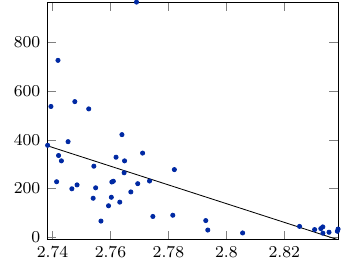}
	& \includegraphics[height=2.2cm]{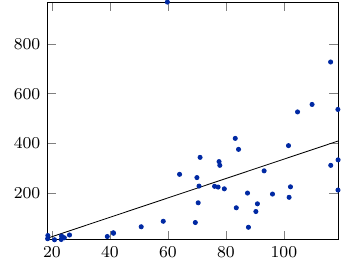}
	& \includegraphics[height=2.2cm]{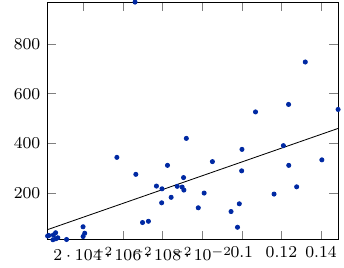}
	\\	% mutliple	   height=2.4cm 
	  \includegraphics[height=2.2cm]{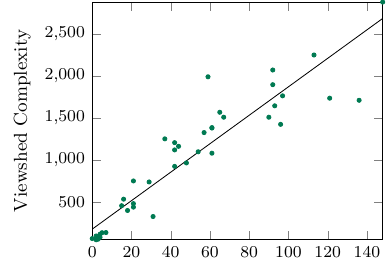}
	& \includegraphics[height=2.2cm]{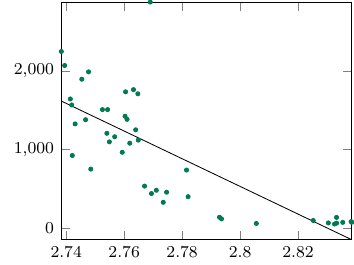}
	& \includegraphics[height=2.2cm]{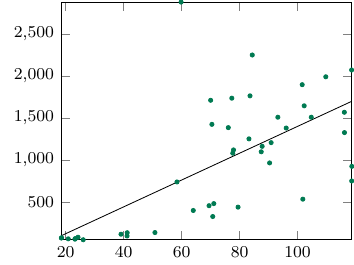}
	& \includegraphics[height=2.2cm]{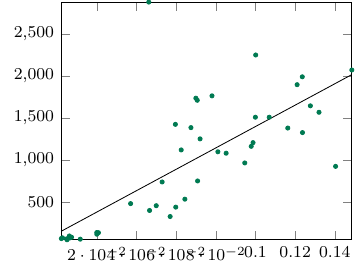}	
  \end{tabularx}
  \caption{The viewshed complexity on a TIN. First row: single random
    viewpoint. Second row: common viewshed of multiple (nine, selected
    from a $3 \times 3$ overlay) random viewpoints.}
  \label{fig:TIN_results}
\end {figure}

For multiple viewpoints, 
%the results comparing the
%topographic attributes to the complexity of the common viewshed of
%multiple viewpoints, 
selected from a $3 \times 3$ overlay grid (refer
to Section~\ref{sec:experiments}), the results are presented in the second row of
Fig.~\ref{fig:TIN_results}.  Again, fractal dimension shows an inverse
behavior.
In contrast, the other three attributes show now a much clearer positive correlation
with viewshed complexity.  
In this case the prickliness shows the
strongest correlation in all but one case (that of viewpoints at
lowest points). In particular when placing the viewpoints at highest
points within the overlay grids the correlation is strong. See also
Fig.~\ref{fig:highest} (left).

\begin{figure}[tb]
  \centering
  % \includegraphics[width=0.3\textwidth]{MeasureUpdates/Plots/DEM_1VP/1VP_Highest_FD_DEM}
  % \includegraphics[width=0.3\textwidth]{MeasureUpdates/Plots/DEM_1VP/1VP_Lowest_FD_DEM}
  %																DEM_
  % \includegraphics[width=0.3\textwidth]{MeasureUpdates/Plots/DEM_1VP/1VP_Highest_TRI_DEM}
  % \includegraphics[width=0.3\textwidth]{MeasureUpdates/Plots/DEM_1VP/1VP_Lowest_TRI_DEM}
  %																DEM_
  % \includegraphics[width=0.3\textwidth]{MeasureUpdates/Plots/DEM_1VP/1VP_Highest_TSI_DEM}
  % \includegraphics[width=0.3\textwidth]{MeasureUpdates/Plots/DEM_1VP/1VP_Lowest_TSI_DEM}
  \hspace{-1em}
  \begin{tabularx}{\linewidth}{cccc}
    % Viewshed Complexity
    ~~~~~Prickliness & ~~~FD & ~~TRI & ~~TSI \\
    \includegraphics[height=2.7cm,valign=t]{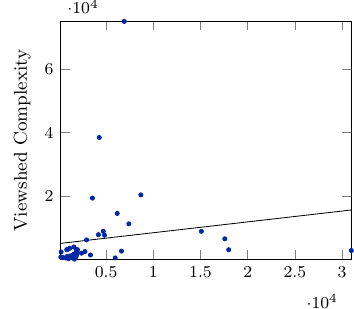}
  & \includegraphics[height=2.5cm,valign=t]{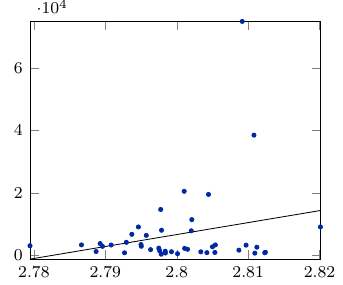}
  & \includegraphics[height=2.5cm,valign=t]{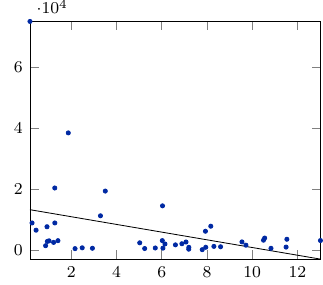}
  & \includegraphics[height=2.5cm,valign=t]{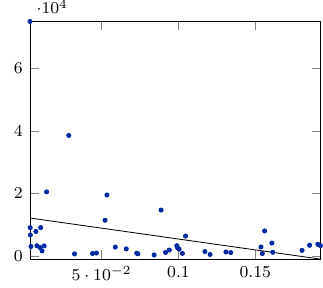}
  \vspace{0.2cm}
  \\
	\includegraphics[height=2.7cm,valign=t]{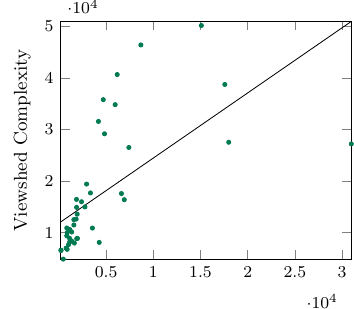}
  & \includegraphics[height=2.5cm,valign=t]{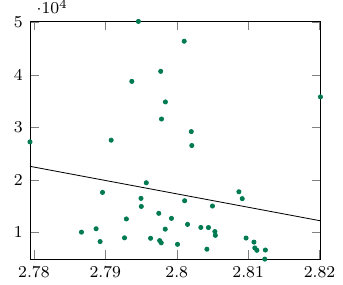}
  & \includegraphics[height=2.5cm,valign=t]{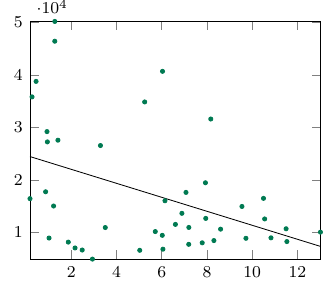}
  & \includegraphics[height=2.5cm,valign=t]{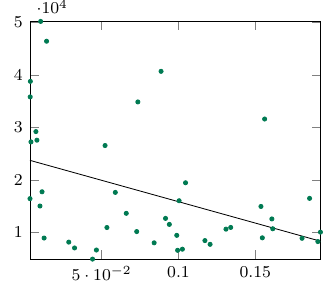}
  \vspace{0.1cm}
  \end{tabularx}
  \caption{The viewshed complexity on a DEM. First row: single random
    viewpoint. Second row: common viewshed of multiple (nine, selected
    from a $3 \times 3$ overlay) random viewpoints. %The prickliness  values are times $10^4$.  
    %All viewshed complexity values and that of prickliness are $\times 10^4$. 
}
  \label{fig:DEM_results}
\end {figure}

\begin{figure}[tb]
  \centering
  \begin{tabular}{cc}
    \quad TIN & \quad DEM \\
  \includegraphics[width=0.3\textwidth]{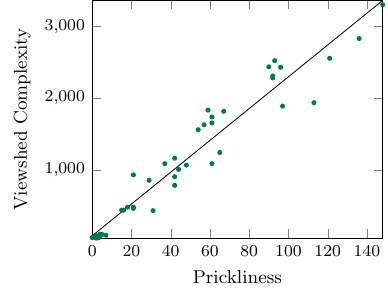}
  &
  \includegraphics[width=0.28\textwidth]{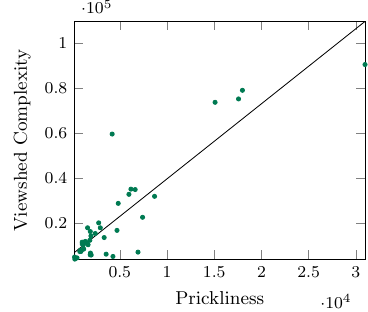}
  \end{tabular}
  \caption{The complexity of the common viewshed of nine viewpoints
    placed using the highest strategy on a TIN (left) or on a DEM
    (right).}
  \label{fig:highest}
\end{figure}
%   % \includegraphics[width=0.3\textwidth]{plots/PricklyPlots/Lowest/Nine_50_Lowest_Prickly_TIN}
%   \includegraphics[width=0.3\textwidth]{plots/PricklyPlots/Random/Nine_50_Random_Prickly_TIN}
%   \includegraphics[width=0.3\textwidth]{MeasureUpdates/Plots/TIN_9VP/9VPs_Random_FD_50_TIN}
%   \includegraphics[width=0.3\textwidth]{MeasureUpdates/Plots/TIN_9VP/9VPs_Random_TRI_50_TIN}
%   \includegraphics[width=0.3\textwidth]{MeasureUpdates/Plots/TIN_9VP/9VPs_Random_TSI_50_TIN}

%   % \includegraphics[width=0.3\textwidth]{MeasureUpdates/Plots/TIN_9VP/9VPs_Highest_FD_50_TIN}
%   % \includegraphics[width=0.3\textwidth]{MeasureUpdates/Plots/TIN_9VP/9VPs_Lowest_FD_50_TIN}
%   \includegraphics[width=0.3\textwidth]{MeasureUpdates/Plots/TIN_9VP/9VPs_Random_FD_50_TIN}
%   %																	9VP/9VPs
%   % \includegraphics[width=0.3\textwidth]{MeasureUpdates/Plots/TIN_9VP/9VPs_Highest_TRI_50_TIN}
%   % \includegraphics[width=0.3\textwidth]{MeasureUpdates/Plots/TIN_9VP/9VPs_Lowest_TRI_50_TIN}
%   \includegraphics[width=0.3\textwidth]{MeasureUpdates/Plots/TIN_9VP/9VPs_Random_TRI_50_TIN}
%   %																	9VP/9VPs
%   % \includegraphics[width=0.3\textwidth]{MeasureUpdates/Plots/TIN_9VP/9VPs_Highest_TSI_50_TIN}
%   % \includegraphics[width=0.3\textwidth]{MeasureUpdates/Plots/TIN_9VP/9VPs_Lowest_TSI_50_TIN}
%   \includegraphics[width=0.3\textwidth]{MeasureUpdates/Plots/TIN_9VP/9VPs_Random_TSI_50_TIN}
%   \caption{The complexity of the common viewshed of nine viewpoints on
%     a TIN compared to the topographic attributes.}
%   \label{fig:Multiple_Tin}
% \end {figure}

\paragraph{DEM results} 
In this case, the results for placing a single viewpoint and placing
the viewpoints in a $3 \times 3$ grid appear somewhat similar. Since
the results for the $3 \times 3$ grid are more pronounced, we focus on
those results.
% most notable difference appears between the three different ways to place viewpoints.

For viewpoints placed randomly, 
the first row of plots in Fig.~\ref{fig:DEM_results} corresponds to the complexity of a single viewshed,
whereas the second row corresponds to the complexity of the common
viewshed of nine viewpoints selected from a $3 \times 3$ overlay.
The correlation values obtained in this setting are shown in the two rightmost columns of Table~\ref{table:combined}.
In contrast to TINs, for DEMs all measures show no to weak correlation
values, even though prickliness obtains the highest correlation coefficient in the case of multiple viewpoints (0.63). 

For viewpoints placed at lowest points the correlation values are very low, both for single and multiple viewpoints (the scatter plots are shown in Fig.~\ref{fig:DEM_results_lowest} in Appendix~\ref{sec:app}).

In the case of viewpoints placed at highest points, correlation results are still weak for all measures except for prickliness (see the scatter plots in Fig.~\ref{fig:DEM_results_highest} in Appendix~\ref{sec:app}). Here prickliness shows a moderate correlation for single viewpoints, but a very high correlation for multiple viewpoints (coefficient 0.90). 
However, the scatter plot (Fig.~\ref{fig:highest} right), where most of the plot points lie in a small region of the plot, suggests that this value may not be very  meaningful.

Finally, it should also be noted that prickliness is not the only attribute that has very different behavior between TINs and DEMs; indeed, the other three attributes also show a very large variation between these two types of terrain models.

\section{Discussion}
\label{section:discussion}

The experimental results for TINs confirm our hypotheses. 
We can see a clear correlation between the viewshed complexity and the prickliness,  especially when multiple viewpoints are
placed on the highest points.
In contrast, this is not evident for the other three topographic attributes considered.
The terrain ruggedness index (TRI) and terrain shape index (TSI)
show some very weak positive correlation, but not as strong as prickliness. 
This could be explained by the fact that
{TRI} and {TSI} only consider a fixed neighborhood around each point, making them local measures unable to capture the whole viewshed complexity. 
Indeed, a small (local) obstruction can be enough to significantly alter the value for any of these attributes.
The fractal dimension (FD) seems to be even worse at predicting viewshed complexity.
Unlike TRI and TSI, this topographic attribute considers the
variability within an area of the terrain as opposed to a
fixed-radius neighborhood. 
Taking a closer look at the FD values for both
terrain datasets shows a minimal variation, with most of them being
close to 3.0, which, according to Taud \etal~\cite{Taud2005Fractal},
indicates a nearly-constant terrain. These results seem to indicate that
this measure fails to detect the variation in elevation levels with
the chosen parameters.

The situation for DEM terrains is less clear.
Only for viewsheds originating from the highest
points we see a strong correlation between prickliness and viewshed complexity.
When the viewpoints are placed at the lowest points of the
DEM terrains, the correlation disappears.  
Since the prickliness measures the amount of
peaks in the terrain in all possible (positive) directions, 
this
means that when a viewpoint is placed at the highest elevation and the
viewshed gets split up by the protrusions (which seem to be accurately
tracked by prickliness), there is a strong correlation.  However, when
the viewpoints are placed at the lowest points, the viewsheds become
severely limited by the topography of the terrain surrounding them.
Even when placing multiple viewpoints, these viewsheds do not
seem to encounter enough of the protrusions that are detected by the
prickliness measure for viewpoints placed at high points.  

One possible explanation for the difference between the results on the TIN
and DEM terrains for prickliness could be attributed to the difference in resolution between the DEMs and TINs used.
The DEMs used consisted of 1.68M cells of 10m size, while the TINs---generated with an error tolerance of 50m---had 1125 vertices on average.
While it would have been interesting to use a higher resolution TIN,
this was not possible due to the high memory usage of the prickliness algorithm.
Another possible explanation for the mismatch between the results for TINs and DEMs may be on the actual definition of prickliness, which is more natural for TINs than for DEMs.
Indeed, it can be seen in the results that the prickliness values for DEMs are much higher than for TINs, which could indicate that the definition is too sensitive to small terrain irregularities.

%  The main reason why prickliness performs so well %compared to the other
%measures may be that it goes beyond measure local differences in
%elevation, which may have small impact in visibility.  For example, if
%there is a column in front of the viewpoint, the viewshed will be
%split regardless of its height. Thus measuring only the elevation
%difference could paint the wrong picture of what is actually affecting
%the viewshed's complexity.  This gives prickliness an advantage
%because it detects the protrusions of the whole terrain without
%considering their elevation.

%When working with TIN terrains of similar complexity, TRI, and TSI can
%be used to get a quick (but rough) idea of the viewshed
%complexity. However, the correlation values of these measures
%deteriorate when increasing the resolution of the terrain.

%Based on our results, it is recommended to use prickliness as a measure to predict the complexity of the viewsheds in these use cases.

\section{Conclusion}
\label{section:conclusion}
We established that prickliness is a reasonable measure of potentially high viewshed complexity, at least for TINs, confirming our theoretical results.
Moreover, prickliness shows a much clearer correlation with viewshed complexity than the three other terrain attributes considered.

One aspect worth further investigation is its correlation for DEMs, which seems to be much weaker. One explanation for this might be that the definition of prickliness is more natural for TINs than for DEMs, but there are several other possible explanations, and it would be interesting future work to delve further into this phenomenon. Having established that
prickliness can be a useful terrain attribute, it remains to improve its computation time, so it can be applied to larger terrains in
practice.

Finally, during our work we noticed that several of the terrain
attributes are defined locally, and are parameterized by some
neighborhood size. Following previous work, we aggregated these local
measures into a global measure by averaging the measurements. It may
be worthwhile to investigate different aggregation methods as
well. This also leads to a more general open question on how to
``best'' transform a local terrain measurement into a global one.

%%%% R: Commented out due to lack of space
\section* {Acknowledgments}

The authors would like to thank Jeff Phillips for a stimulating discussion that, years later, led to the notion of prickliness. They would also like to thank Ramesh K. Jallu for his work on previous versions of this paper, and Hans Raj Tiwary for key observations in Sections~\ref{sub:Algorithm_for_TINs} and~\ref{sub:Lower_bound_2.5D}.

%A.A., R.J., and M.S. 
A.A. and M.S.
were supported by the Czech Science Foundation, grant number GJ19-06792Y.  M.L. was partially supported by the Netherlands Organization for Scientific Research (NWO) under project no. 614.001.504. R.S. was supported by project PID2023-150725NB-I00 funded by MICIU/ AEI/ 10.13039/ 501100011033. %
% and Gen. Cat. DGR 2017SGR1640. 
This project has received funding from the European Union's Horizon 2020 research and innovation programme under the Marie Sk\l{}odowska-Curie grant agreement No 734922. The work was partially done while A.A. and M.S. were affiliated with Institute of Computer Science of the Czech Academy of Sciences, with institutional support RVO:67985807.

%%
%% Bibliography
%%

%% Please use bibtex, 
\bibliographystyle{splncs04}
\bibliography{refs}
\newpage

\appendix
\begin{appendix}
\section{Additional plots} \label{sec:app}

In this section, we provide the scatter plots for the placement strategies of choosing ``high'' points and for choosing ``low'' points. Plots for TINs are displayed in Fig.~\ref{fig:TIN_results_highest} and~\ref{fig:TIN_results_lowest}, while plots for DEMs are shown in Fig.~\ref{fig:DEM_results_highest} and~\ref{fig:DEM_results_lowest}.

\begin{figure}[h]
  \hspace{-1em}
  \begin{tabularx}{\linewidth}{cccc}
    ~~~~~Prickliness & ~~~FD & ~~TRI & ~~TSI \\
      \includegraphics[height=2.2cm]{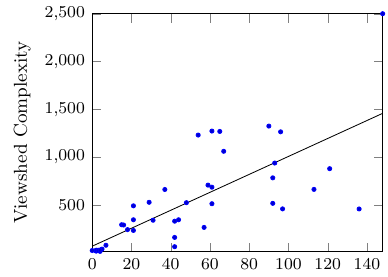}
	& \includegraphics[height=2.2cm]{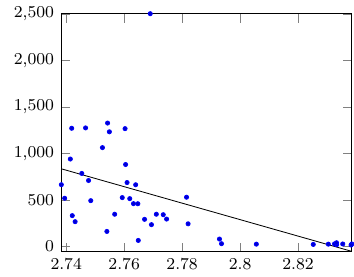}
	& \includegraphics[height=2.2cm]{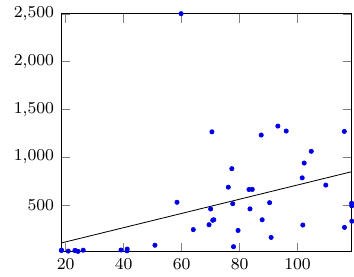}
	& \includegraphics[height=2.2cm]{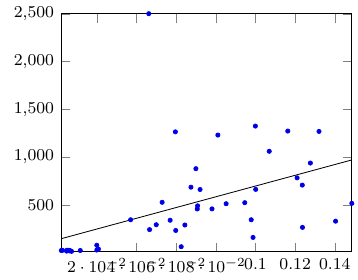}
	\\	 
	  \includegraphics[height=2.2cm]{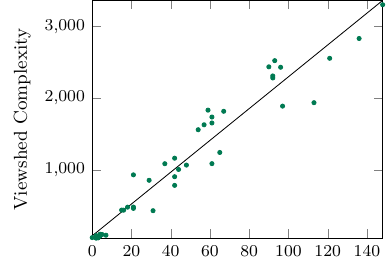}
	& \includegraphics[height=2.2cm]{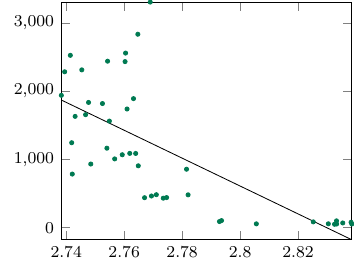}
	& \includegraphics[height=2.2cm]{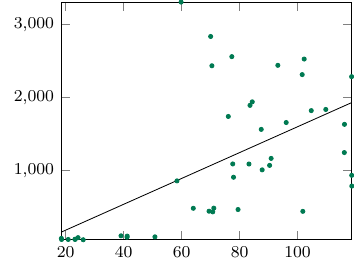}
	& \includegraphics[height=2.2cm]{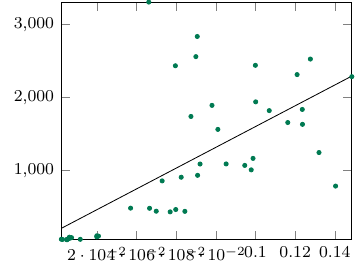}	
  \end{tabularx}
  \caption{The viewshed complexity on a TIN. First row: single highest
    viewpoint. Second row: common viewshed of multiple (nine, selected
    from a $3 \times 3$ overlay) highest viewpoints.}
  \label{fig:TIN_results_highest}
\end {figure}

\begin{figure}[h]
  \hspace{-1em}
  \begin{tabularx}{\linewidth}{cccc}
    ~~~~~Prickliness & ~~~FD & ~~TRI & ~~TSI \\
      \includegraphics[height=2.2cm]{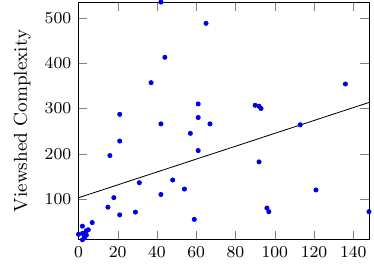}
	& \includegraphics[height=2.2cm]{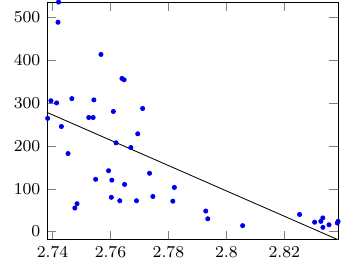}
	& \includegraphics[height=2.2cm]{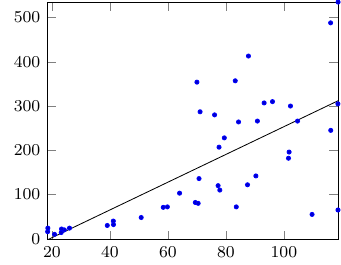}
	& \includegraphics[height=2.2cm]{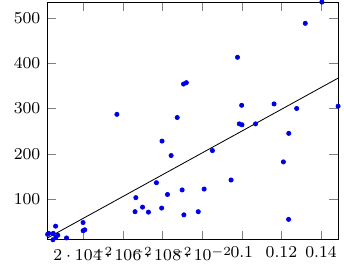}
	\\	 
	  \includegraphics[height=2.2cm]{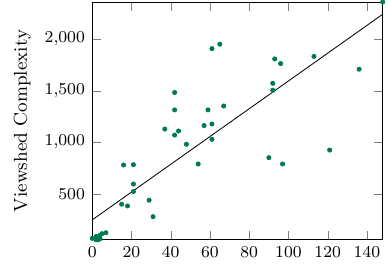}
	& \includegraphics[height=2.2cm]{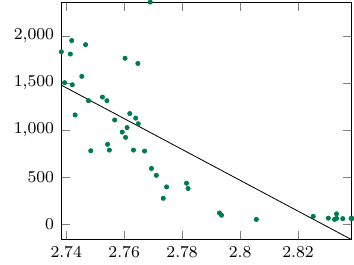}
	& \includegraphics[height=2.2cm]{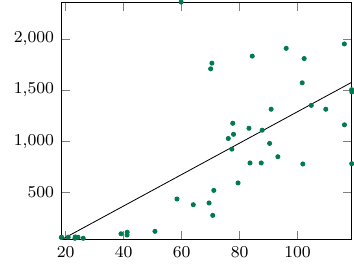}
	& \includegraphics[height=2.2cm]{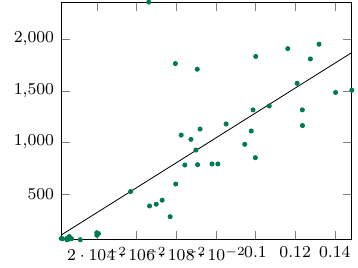}	
  \end{tabularx}
  \caption{The viewshed complexity on a TIN. First row: single lowest
    viewpoint. Second row: common viewshed of multiple (nine, selected
    from a $3 \times 3$ overlay) lowest viewpoints.}
  \label{fig:TIN_results_lowest}
\end {figure}

\begin{figure}[tb]
  \centering
   \hspace{-1em}
  \begin{tabularx}{\linewidth}{cccc}
    % Viewshed Complexity
    ~~~~~Prickliness & ~~~FD & ~~TRI & ~~TSI \\
    \includegraphics[height=2.7cm,valign=t]{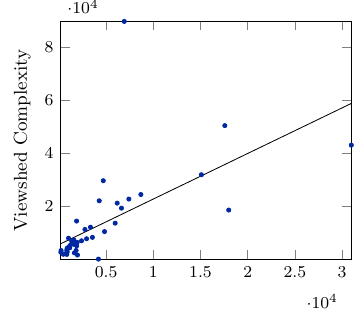}
  & \includegraphics[height=2.5cm,valign=t]{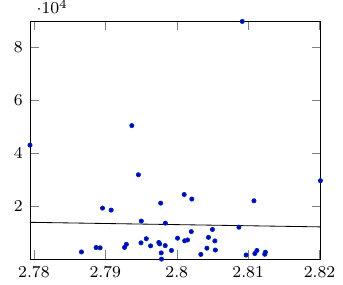}
  & \includegraphics[height=2.5cm,valign=t]{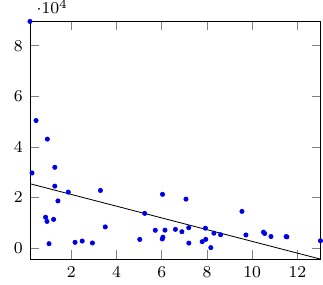}
  & \includegraphics[height=2.5cm,valign=t]{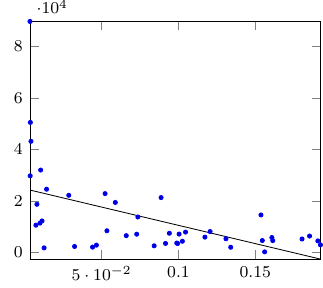}
  \vspace{0.2cm}
  \\
	\includegraphics[height=2.7cm,valign=t]{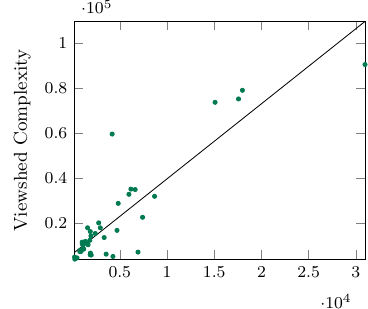}
  & \includegraphics[height=2.5cm,valign=t]{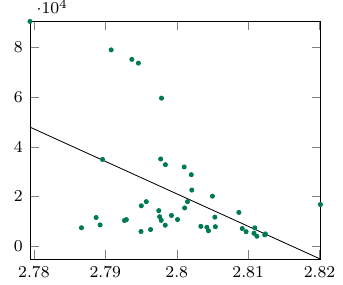}
  & \includegraphics[height=2.5cm,valign=t]{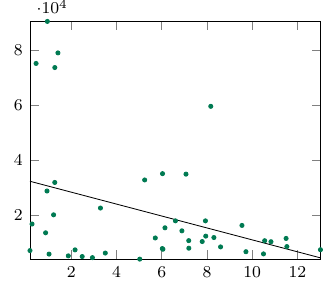}
  & \includegraphics[height=2.5cm,valign=t]{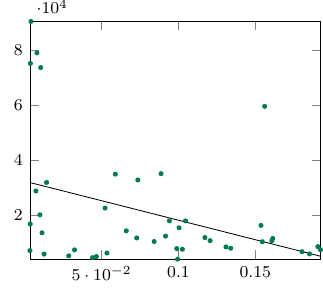}
  \vspace{0.2cm}
  \end{tabularx}
  \caption{The viewshed complexity on a DEM. First row: single highest
    viewpoint. Second row: common viewshed of multiple (nine, selected
    from a $3 \times 3$ overlay) highest viewpoints. 
}
  \label{fig:DEM_results_highest}
\end {figure}

\begin{figure}[tb]
  \centering
   \hspace{-1em}
  \begin{tabularx}{\linewidth}{cccc}
    % Viewshed Complexity
    ~~~~~Prickliness & ~~~FD & ~~TRI & ~~TSI \\
    \includegraphics[height=2.7cm,valign=t]{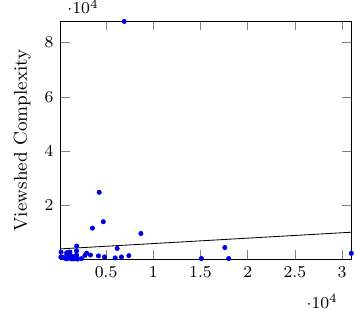}
  & \includegraphics[height=2.5cm,valign=t]{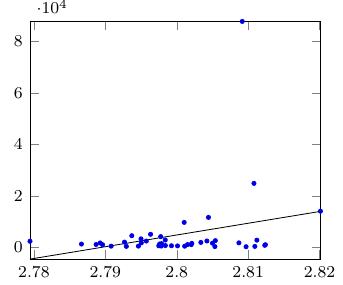}
  & \includegraphics[height=2.5cm,valign=t]{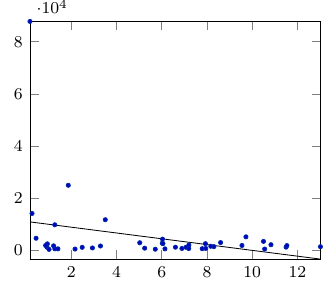}
  & \includegraphics[height=2.5cm,valign=t]{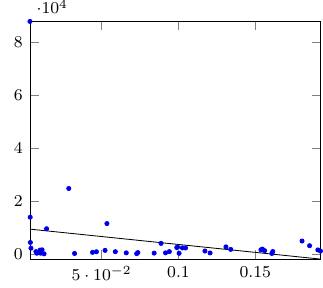}
  \vspace{0.2cm}
  \\
	\includegraphics[height=2.7cm,valign=t]{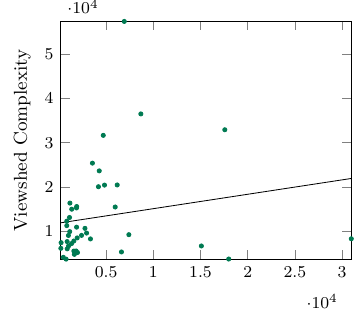}
  & \includegraphics[height=2.5cm,valign=t]{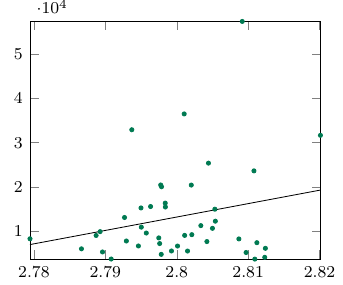}
  & \includegraphics[height=2.5cm,valign=t]{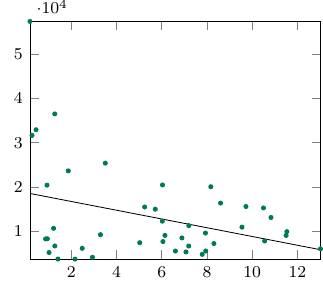}
  & \includegraphics[height=2.5cm,valign=t]{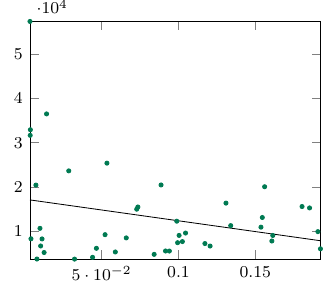}  \vspace{0.2cm}
  \end{tabularx}
  \caption{The viewshed complexity on a DEM. First row: single lowest
    viewpoint. Second row: common viewshed of multiple (nine, selected
    from a $3 \times 3$ overlay) lowest viewpoints. 
}
  \label{fig:DEM_results_lowest}
\end {figure}
\end{appendix}

\end{document}